\documentclass[11pt,reqno,a4paper]{amsart} 
\usepackage{amsmath, amstext, amssymb, amsthm}
\usepackage[T1]{fontenc}
\usepackage[latin1]{inputenc}
\usepackage{multirow} 
\usepackage{wrapfig}
\usepackage{subfig}
\usepackage{epsf}
\usepackage{booktabs}

\usepackage{enumerate}
\usepackage{algorithmicx}
\usepackage{algpseudocode}
\usepackage{url}

\usepackage{colortbl}
\usepackage{color}
\usepackage{array}
\definecolor{shg}{rgb}{0.55,0.55,0.55}
\definecolor{hg}{rgb}{0.65,0.65,0.65}
\definecolor{mg}{rgb}{0.75,0.75,0.75}
\definecolor{dg}{rgb}{0.85,0.85,0.85}
\definecolor{wh}{rgb}{1,1,1}
\definecolor{sdg}{rgb}{0.95,0.95,0.95}

\def\hbm{\mbox{}\hfill$\boxdot$}

\usepackage{graphicx}
\usepackage{psfrag}
\usepackage{chemarrow}
\usepackage{rotating}
\usepackage{lscape}

\def\R{\ensuremath{I\!\!R}}
\def\Rnn{\ensuremath{I\!\!R_{\geq 0}}}
\def\Rp{\ensuremath{\R_{> 0}}}

\newtheorem{definition}{Definition}[section]
\newtheorem{remark}[definition]{Remark}
\newtheorem{theorem}[definition]{Theorem}
\newtheorem{lem}[definition]{Lemma}
\newtheorem{prop}[definition]{Proposition}
\newtheorem{coro}[definition]{Corollary}

\newtheorem{example}[definition]{Example}

\DeclareMathOperator{\sign}{sign}
\DeclareMathOperator{\im}{im}
\DeclareMathOperator{\col}{col}
\DeclareMathOperator{\diag}{diag}
\DeclareMathOperator{\supp}{supp}

\newcommand{\brac}[1]{\ensuremath{\left(#1\right)}}
\newcommand{\sig}[1]{\ensuremath{\sign\brac{#1}}}
\newcommand{\uu}[1]{\ensuremath{\underline{#1}}}

\newcommand{\Si}[1]{\ensuremath{S^{\left(#1\right)}}}
\newcommand{\Yi}[1]{\ensuremath{{\mathcal{Y}}^{\left(#1\right)}}}
\newcommand{\Mi}[1]{\ensuremath{M^{\brac{#1}}}}
\newcommand{\Zi}[1]{\ensuremath{Z^{\brac{#1}}}}
\newcommand{\Ei}[1]{\ensuremath{E^{\brac{#1}}}}
\newcommand{\ei}[1]{\ensuremath{e^{\brac{#1}}}}
\newcommand{\sigi}[1]{\ensuremath{\sigma^{\brac{#1}}}}
\newcommand{\PPi}[1]{\ensuremath{\Pi^{\brac{#1}}}}
\newcommand{\e}[1]{\ensuremath{e^{\brac{#1}}}}

\newcommand{\Yo}[1]{\ensuremath{Y_0\brac{#1}}}
\newcommand{\rbrac}[1]{\left[#1\right]}

\newcommand{\ri}[3]{\ensuremath{r^{\brac{#1}}}\brac{#2,\, #3}}

\newcommand{\matMi}[1]{\ensuremath{\mathcal{M}^{\brac{#1}} }}

\begin{document}
\captionsetup[subfloat]{farskip=0pt} 
\title{Multistationarity in sequential distributed multisite
  phosphorylation networks}

\author{Katharina Holstein, Dietrich Flockerzi and Carsten Conradi}
{\address{
    Max-Planck-Institut Dynamik komplexer technischer Systeme, 
    Sandtorstr.\ 1, 39106 Magdeburg, Germany.}}

\email{
  holstein@mpi-magdeburg.mpg.de,flockerzi@mpi-magdeburg.mpg.de,conradi@mpi-magdeburg.mpg.de}

\begin{abstract}
  Multisite phosphorylation networks are encountered in many
  intracellular processes like signal transduction, cell-cycle control or
  nuclear signal integration. In this contribution networks describing the
  phosphorylation and dephosphorylation of a protein at $n$ sites in a
  sequential distributive mechanism are
  considered. Multistationarity (i.e.\ the existence of at least two
  positive steady state solutions of the associated polynomial
  dynamical system) has been {analyzed} and established in several
  contributions. It is, for example, known that there exist values for
  the rate constants where multistationarity occurs. However, nothing
  else is {known} about these rate constants. 

  Here we present a sign condition that is necessary and sufficient
  for multistationarity in $n$-site sequential, distributive
  phosphorylation. We express this sign condition in terms of linear
  systems and show that solutions of these systems define rate
  constants where multistationarity is possible. We then present, for
  $n\geq 2$, a collection of {\em feasible} linear systems and hence
  give a new and independent proof that multistationarity is possible
  for $n\geq 2$. Moreover, our results allow to explicitly obtain
  values for the rate constants where multistationarity is
  possible. Hence we believe that, for the first time, a systematic
  exploration of the region in parameter space where multistationarity
  occurs has become possible.
  One consequence of our work is that, for any pair of steady
  states, the ratio of the steady state concentrations of
  kinase-substrate complexes  equals that of phosphatase-substrate
  complexes.
  
\end{abstract}

\maketitle
\section{Introduction}

Protein phosphorylation and dephosphorylation is frequently
encountered in procaryotic and in eukaryotic cells. Thereby the
majority of proteins is phosphorylated at more than one
phosphorylation site (with most procaryotic proteins having $n\leq 7$
phosphorylation sites and eukaryotic proteins generally having a
large number of phosphorylation sites -- certain eukaryotic proteins
have up to $n=150$ sites \cite{MT-09}). Multisite
phosphorylation networks are consequently encountered in many
intracellular processes, including, among others, the following:
(i) signal transduction networks from the family of Mitogen Activated
Kinases (MAPK cascades), where the core module consists of one protein
phosphorylated at a single site and two proteins each phosphorylated
at two sites (see e.g.\ \cite{sig-016,sig-014,sig-051,sig-052}); (ii)
check-points in cell-cycle control, for example the control of the
G1/S transition in {budding} yeast involving the protein \emph{Sic1}
with $n=9$ (see e.g.\ \cite{cyc-022,cyc-004,jjt-08}) and (iii) nuclear
signal integration, for example, in mammalian cells involving
members of the {\em NFAT} transcription factor family with $n=13$ (see
e.g.\ \cite{MTC-09,NFAT-007,NFAT-002,NFAT-003,TT03}).

In this contribution we focus on networks describing the
phosphorylation and dephosphorylation of a single protein $A$ at $n$
sites in a sequential distributive mechanism (cf.\
Figure~\ref{fig-1} and, for example, \cite{ptm-001}). Such a mechanism
is able to cover the essential aspects of proteins like the
aforementioned \emph{Sic1} or the members of the \emph{NFAT} family
\cite{CS07,sig-041}.
\begin{figure}[ht!]
  \centering
  \includegraphics[width=1\textwidth]{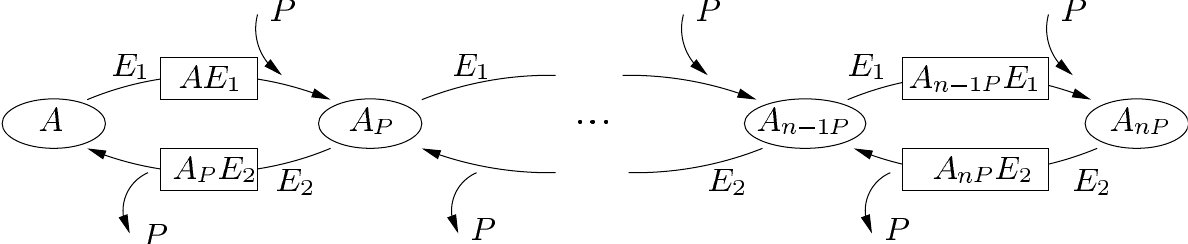}
  \caption{
    Network describing the sequential distributive
    phosphorylation and dephosphorylation of protein $A$ at
    $n$-sites by a kinase $E_1$ a phosphatase $E_2$. The
    phosphorylated forms of $A$ are denoted by the subscript $nP$
    (denoting the number of phosphorylated sites).
  }
  \label{fig-1}
\end{figure}

Due to the biological significance of sequential distributive
phosphorylation, mathematical models have been studied in a variety of
publications. Several publications deal with double
phosphorylation (i.e.\ $n=2$)  and establish multistationarity
\cite{cc-08,cc-05,fein-018,ptm-007} (i.e.\ the existence of at least two
positive steady state solutions, cf.\ Definition~\ref{defimulti}) and
bistability \cite{NM-04} (i.e.\ the existence of two (asymptotically)
stable steady state solutions). Networks with $n$ phosphorylation sites
have been studied, for example, in \cite{CS07,sig-041}, where
bistability is demonstrated numerically and in \cite{JG05}, where it
is argued that even though bistability is in principle possible, it is
unlikely to occur in-vivo. 

The set of positive steady states of networks with $n$ phosphorylation 
sites has been studied algebraically in \cite{ptm-006} and in
\cite{ToricRN}. In \cite{ptm-006} the set of positive steady
states of the more general class of so called post translational
modification networks (with $n$-site sequential distributive
phosphorylation as a special case) is {analyzed}. Such networks admit a
rational parameterization of the set of positive steady states. And in
\cite{ToricRN} it is shown that for $n$-site sequential distributive
phosphorylation this set is defined by a binomial ideal. Hence $n$-site 
sequential distributive phosphorylation is an 
instance of a Chemical Reaction System with toric steady states. In
\cite{ToricRN} a condition for multistationarity in these reaction
systems is given in terms of a sign condition on elements of two linear
subspaces. Based on the results presented in \cite{fein-043,cc-08},
this condition can be restated in terms of linear systems: to decide
the sign condition, a large set of linear systems has to be examined
and multistationarity is possible, if and only if at least one of
these is feasible (cf.\ \cite[Theorem~5.5]{ToricRN} and
\cite[Lemma~2~\&{~}Theorem~2]{fein-043}).

The number of steady states in general $n$-site sequential
distributive phosphorylation has been studied in
\cite{multi-001}. There upper and lower bounds have been established:
there exist parameter values such that, for $n$ even (odd), there are
at least $n+1$ ($n$) steady states. And for all (positive) parameter
values there are at most $2\ n-1$ steady states. Hence
multistationarity has been established for sequential distributive
phosphorylation in \cite{multi-001}. However, no information other
than existence is given about the parameter values where
multistationarity is possible. In this contribution we address this
question: we present for every $n\geq 2$ a collection of {\em feasible
  linear systems} and show that every solution of one of those systems
defines parameter values where multistationarity is possible (together
with two positive steady states as witness). Thus these results not
only constitute a new and independent proof that multistationarity is
possible for $n\geq 2$, they additionally enable, for the first time,
a systematic exploration of the region in parameter space where
multistationarity is possible. Due to the ubiquity of multisite
phosphorylation this may be of potential interest to researchers
working in many fields of (quantitative) biology. For the purpose of
finding parameter values in biologically meaningful ranges we will
start with this exploration in \cite{kh-13}.

We arrive at our main result on the basis of our previous work from
\cite{fein-043,cc-08} and \cite{ToricRN}: as in the aforementioned
references we derive a sign condition that is necessary and sufficient
for the existence of multistationarity (Theorem~\ref{theo:poly_sol} \&
\ref{theo:multi_linear_problem}). In contrast to
\cite{fein-043,ToricRN} our condition is formulated in terms of
the sign patterns generated by two linear subspaces (given as the
image of two matrices defined in eqns.\ (\ref{eq:S_direct}) \&
(\ref{eq:def_Mn})). We then exploit the fact that every sign pattern
uniquely defines a linear system to restate the sign condition in
terms of these linear systems: multistationarity is possible, if and
only if at least one out of 
$\frac{1}{2}\, \big(3^{\brac{3n+3}}-1\big)$
or
$2^{\brac{3n+2}}$
linear systems is feasible
depending on whether some or all components of the two 
witness steady
state
vectors have to differ to qualify for multistationarity,
cf.~Definition~\ref{defimulti} and Remark~\ref{rem:matro-condi}. 
For the latter case, we present formulae to construct 
all sign patterns defining
{\em feasible} linear systems
(cf.~Theorem~\ref{theo:feasible_patterns}) and show that every solution
of one of these linear systems defines parameter  values where
multistationarity occurs with two positive steady states differing in
every component.
We show that there are ${2\, (n-1)(n+2)}$ 
such sign patterns (cf.~Proposition~\ref{lem:nr_of_elements}).

This paper is organized as follows: Section~\ref{sec:notation}
introduces some of the basic notations. 
In
Section~\ref{sec:models} a mass action network is derived from
Fig.~\ref{fig-1} and the associated dynamical system is
defined
along with a formal definition of multistationarity.
In Section~\ref{sec:results} the main results address
multistationarity in terms of sign patterns  defining
feasible linear systems
(cf.~Theorem~\ref{theo:poly_sol} \& \ref{theo:feasible_patterns}).
The proofs of two of these results have
been relegated to separate Sections~\ref{sec:proof_poly_sol} \&
\ref{sec:proof_theo_patterns} for easier reading.

\vspace*{5mm}\section{Notation}
\label{sec:notation}

We use the symbol $\R^m$ to denote Euclidian $m$-space, the symbol
$\Rnn^m$ to denote the nonnegative orthant and $\Rp^m$ to denote the
interior of the nonnegative orthant. Vectors are considered as column
vectors and, for convenience, usually displayed as row vectors using
$^T$ to denote the transpose. For example,
$x\in \R^m$ will usually be displayed as $\brac{x_1,\, \ldots,\,    x_m}^T$.

For positive vectors $x\in\Rp^m$ we use the shorthand notation $\ln x$
to denote
\begin{subequations}
  \begin{align}
    \label{eq:def_ln_x}
    \ln x &:= \brac{\ln x_1,\, \ldots,\, \ln x_m}^T\, \in \, \R^m. \\
    \intertext {Similarly, for $x\in\R^m$, we use $e^x$ to denote}
    e^x &:= \brac{e^{x_1},\, \ldots,\, e^{x_m}}^T\, \in \, \Rp^m
    \intertext{and, for $x\in\R^m$ with $x_i\neq 0$, $i=1$, $\ldots$, $m$,}
    x^{-1} &:= \brac{\frac{1}{x_1},\, \ldots,\, \frac{1}{x_m}}^T\, \in \, \R^m\, .
  \end{align}
Finally, $x^y$ with $x$, $y\in\Rnn^m$ will be defined by
\begin{equation}\label{eq:def_xy}
  x^{y} := \prod_{i=1}^m x_i^{y_i}\, \in \, \Rnn^m\, .
\end{equation}\end{subequations}
The set of real $(p\times q)$-matrices with $p$ rows
and $q$ columns will be denoted by $\R^{p\times q}$. For
$p$-dimensional column vectors $s_{(j)}=(s_{1j},...,s_{pj})^T\in
\R^p$, $j=0,1,...,q$, we denote the $p(q+1)$-dimensional column
vector $(s_{10},...,s_{p0},s_{11},...,s_{p1},\, .......\,
,s_{1q},...,s_{pq})^T$ by 
\begin{equation}\label{eq:def_col}
  \col(s_{(0)},...,s_{(q)})\, \in \, \R^{p(q+1)} \ \mbox{ with the matrix }
  (s_{(0)},...s_{(q)})\, \in \, \R^{p\times (q+1)}\ .
\end{equation}

\vspace*{5mm}\section{Mathematical models of $n$-site  sequential distributive
  phosphorylation}
\label{sec:models}

Here we first introduce the notation used to describe dynamical systems
defined by mass action networks by means of the example network
depicted in Fig.~\ref{fig:n=1_exa}. Then we discuss the mass action network derived
from Fig.\ref{fig-1} and present the dynamical system defined by
this mass action network.

\begin{figure}[!htb]
  \centering
  \subfloat[$1$-site phosphorylation\label{fig:1-site}]{
    \includegraphics[width=.4\linewidth]{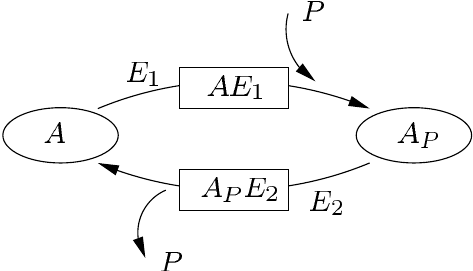}
  }
  \hfill
  \subfloat[Mass action network \label{eq:net_phospho_1}]{
    \begin{minipage}[b]{.4\linewidth}
      \begin{gather*}
        \overset{x_1}{E_1} + \overset{x_2}{A}
        \autorightleftharpoons{$k_1$}{$k_2$} \overset{x_4}{AE_1}
        \autorightarrow{$k_3$}{} \overset{x_1}{E_1} + \overset{x_5}{A_P}
        \\ 
        \overset{x_3}{E_2} +
        \overset{x_5}{A_P}\autorightleftharpoons{$l_1$}{$l_2$}
        \overset{x_6}{A_PE_2} \autorightarrow{$l_3$}{}
        \overset{x_3}{E_2} + \overset{x_2}{A} 
      \end{gather*}
      \vspace*{2ex}
    \end{minipage}
  }
  \caption{Phosphorylation of protein $A$ at a single phosphorylation
    site: process diagram (A) and mass action network (B).}
  \label{fig:n=1_exa}
\end{figure}

\subsection{Dynamical systems defined by a single phosphorylation network}
\label{sec:dyn_sys_1s}

The network depicted in Fig.~\ref{fig:n=1_exa} describes the
phosphorylation of a protein $A$ at a single site by the kinase $E_1$
and the dephosphorylation of mono-phosphorylated protein $A_P$ by the
phosphatase $E_2$ (cf.\
Fig.~\ref{fig:1-site}). Figure~\ref{eq:net_phospho_1} contains a mass
action network describing this process, where chemical reactions are
denoted by arrows pointing from the species that are consumed in a
reaction (the educts) to the species that are produced (the products)
and $k_i$ ($l_i$) denote rate constants.  The network consists of 6
species, the kinase $E_1$ (with concentration variable $x_1$), the
protein $A$ ($x_2$), the phosphatase $E_2$ ($x_3$), the enzyme
substrate complexes $AE_1$ ($x_4$) and $A_PE_2$ ($x_6$) and the
mono-phosphorylated $A_P$ ($x_5$). The network is composed of 6
reactions, two describing the reversible binding of protein and kinase
$E_1 + A {\rightleftharpoons} AE_1$, one the irreversible
release of mono-phosphorylated protein and kinase $AE_1 
{\to} E_1 + A_P$, two the reversible binding of
mono-phosphorylated protein and phosphatase $E_2 + A_P
{\rightleftharpoons} A_PE_2$ and one the irreversible release 
unphosphorylated protein and phosphatase $A_PE_2{\to}
E_2 + A$.

To every reaction we associate a reaction vector describing the
stoichiometry of that reaction. For example in the reaction
\begin{displaymath}
  E_1 + A \autorightarrow{$k_1$}{} AE_1
\end{displaymath}
one unit of $A$ ($x_2$) and $E_1$ ($x_1$) are consumed and one unit of
$A E_1$ ($x_4$) is produced, which yields the vector
\begin{displaymath}
  \left(-1,\, -1,\, 0,\, 1,\, 0,\, 0\right)^T\ .
\end{displaymath}
All in all one obtains six reaction vectors, that we collect as
column vectors of the stoichiometric matrix
\begin{displaymath}
  \Si{1}= \left[
    \begin{array}{cccccc}
      -1 & \phantom{-}1   & \phantom{-}1   & \phantom{-}0   &
      \phantom{-}0   & \phantom{-}0   \\
      -1 & \phantom{-}1   & \phantom{-}0   & \phantom{-}0   &
      \phantom{-}0   & \phantom{-}1   \\ 
      \phantom{-}0 & \phantom{-}0   & \phantom{-}0   & -1  &
      \phantom{-}1   & \phantom{-}1   \\       
      \phantom{-}1 & -1  & -1 & \phantom{-}0   & \phantom{-}0   &
      \phantom{-}0   \\  
      \phantom{-}0 & \phantom{-}0   & \phantom{-}1   & -1  &
      \phantom{-}1   & \phantom{-}0   \\ 
      \phantom{-}0 & \phantom{-}0   & \phantom{-}0   & \phantom{-}1
      &-1   & -1
    \end{array}
  \right]\ ,
\end{displaymath}
where the superscript $^{\brac{1}}$ is used to denote the
stoichiometric matrix of a 1-site phosphorylation network. Columns 1
and 2 of \Si{1} correspond to the reversible binding of $A$ and $E_1$,
column 3 to the irreversible release of $A_P$ and $E_1$, columns 4 and
5 to the reversible binding of $A_P$ and $E_2$ and column 6 to the
irreversible release of $A$ and $E_2$.

Using mass action kinetics the reaction rate $r_i$ (i.e.\ the speed of
reaction $i$) is proportional to the product of the concentrations of
the educts, for example, the rate $r_1=k_1\, x_1\, x_2$ for the
reaction $E_1+A \to AE_1 $. Similarly one obtains $r_2=k_2\, x_4$ for
$AE_1\to  E_1 + A$, $r_3=k_3\, x_4$ for $AE_1\to  E_1 + A_P$,
$r_4=l_1\, x_3\, x_5$ for $E_2 + A_P\to A_PE_2$, $r_5=l_2\, x_6$ for
$A_PE_2 \to E_2 + A_P$ and $r_6=l_3\, x_6$ for $A_PE_2 \to E_2 +
A$. 
With $x=(x_1$,\ldots, $x_6)^T$, $\kappa=(k_1$, \ldots, $k_3$, $l_1$, \ldots,
$l_3)^T$
we collect $r_1$, \ldots, $r_6$ in the vector 
\begin{subequations}\label{eq:14s_ode}
\begin{equation}\label{eq:1rs_ode}
  \ri{1}{\kappa}{x} = \left(k_1\, x_1\, x_2,\, k_2\, x_4,\, k_3\,
    x_4,\, l_1\, x_3\, x_5,\, l_2\, x_6,\, l_3\, x_6\right)^T
\end{equation}
and obtain the dynamical system 

\begin{align}
  \label{eq:1s_ode}
  \dot x &= \Si{1}\, \ri{1}{\kappa}{x} \, ,\\
  \intertext{which can be written componentwise as}
  \notag
  \dot{x}_1 & = \left(k_2+k_3\right)x_4 - k_1x_1x_2,\\
  \notag
  \dot{x}_2 & = -k_1x_1x_2 + k_2x_4 + l_3x_6, \\
  \notag
  \dot{x}_3 & = \left(l_2+l_3\right)x_6 - l_1x_3x_5,  \\
  \label{eq:1xis_ode}
  \dot{x}_4 & = - \left(k_2+k_3\right)x_4 + k_1x_1x_2, \\
  \notag
  \dot{x}_5 & = k_3x_4 - l_1x_3x_5 + l_2x_6, \\
  \notag
  \dot{x}_6 & = - \left(l_2+l_3\right)x_6 + l_1x_3x_5.
\end{align}
Note that the vector $\ri{1}{\kappa}{x}$ is fully characterized by the
vector $\kappa$ and the exponents of the monomials. We collect those
in the rate exponent matrix \Yi{1} 
corresponding to the educt complexes (i.e., the \lq left-hand
sides\rq\,  of the reactions): 
\begin{equation}\label{eq:1ys_ode}
  \Yi{1} = \left[
    \begin{array}{cccccc}
      1 & 0 & 0 & 0 & 0 & 0 \\
      1 & 0 & 0 & 0 & 0 & 0 \\
      0 & 0 & 0 & 1 & 0 & 0 \\
      0 & 1 & 1 & 0 & 0 & 0 \\
      0 & 0 & 0 & 1 & 0 & 0 \\
      0 & 0 & 0 & 0 & 1 & 1 
    \end{array}
  \right].
\end{equation}
\end{subequations}
Fig.~\ref{fig:n=1_exa}
does not consider protein synthesis and
degradation, the total amount of protein $A$ is therefore constant,
that is, one has the conservation relation 
\begin{displaymath}
  x_2+x_4+x_5+x_6 = c_1,
\end{displaymath}
where $c_1$ is a constant denoting the total concentration of
$A$. The name conservation relation stems form the fact that the above
sum of concentrations is constant along solutions of
(\ref{eq:1s_ode}):
from (\ref{eq:1xis_ode})
it is easy to see that
$\dot x_ 2\brac{t} + \dot x_4\brac{t} + \dot x_5\brac{t} + \dot
x_6\brac{t} = 0$ and hence $x_ 2\brac{t} +  x_4\brac{t} + x_5\brac{t}
+ x_6\brac{t} = const$. Likewise one has
\begin{displaymath}
  x_1+x_4 = c_2\qquad \text{and}\qquad x_3+x_6 = c_3,
\end{displaymath}
with constants $c_2$ and $c_3$ denoting the total concentration of
$E_1$ and $E_2$, respectively. 
Conservation relations are defined by elements of the left
kernel of the matrix \Si{1}. One obtains, for example, the following
full row-rank matrix 
\begin{displaymath}
  \Zi{1} =\left[
    \begin{array}{cccccc}
      1 & 0 & 0 & 1 & 0 & 0 \\
      0 & 0 & 1 & 0 & 0 & 1 \\
      0 & 1 & 0 & 1 & 1 & 1
    \end{array}
  \right]
\end{displaymath}
with ${\Zi{1}}\, \Si{1} \equiv 0$ where the rows form a basis of the left kernel of $\Si{1}$.

\subsection{The mass action network derived from Figure~\ref{fig-1}}
\label{sec:network}

The mass action network derived from Fig.~\ref{fig-1} (with $n$ an
arbitrary but fixed positive number) has been described in
\cite{ToricRN,multi-001}. Here we use a similar mathematical
description:
the network consists of the following chemical species: the protein
(substrate) $A$ together with $n$ phosphoforms $A_P$, \ldots,
$A_{nP}$; the kinase $E_1$ together with $n$ kinase-substrate
complexes $A\, E_1$, \ldots, $A_{n-1P}\, E_1$ and the phosphatase
$E_2$ together with $n$ phosphatase-substrate complexes $A_{P}\, E_2$,
\ldots, $A_{nP}\, E_2$. Hence there is a total of $3n+3$ species. To
each species, a variable $x_i$ denoting its concentration is
assigned as 
depicted in
Table~\ref{tab:VarAssignment}. We collect all variables in a $(3n+3)$-dimensional vector
\begin{displaymath}
  x := \brac{x_1,\, \ldots,\, x_{3n+3}}^T\ .
\end{displaymath}
As it will turn out, the chosen labeling entails a simple block structure for
the matrices associated to the  dynamical system \eqref{eq:def_rnkx} of the network in Fig.~\ref{fig-1},
cf., for example, the block structure \eqref{eq:Ei} for the generators
of the nonnegative cone in the kernel of the stoichiometric matrix.

\begin{table}
  \centering
  \begin{tabular}{|c|c|c|} \hline
    Phos. \# & Species & Var. \\ \hline
    \multirow{3}{*}{0} & $E_1$ & $x_1$ \\
    & $A$ & $x_2$ \\
    & $E_2$ & $x_3$ \\ \hline
    \multirow{3}{*}{1} & $A\, E_1$ & $x_4$ \\
    & $A_P$ & $x_5$ \\
    & $A_P\, E_2$ & $x_6$ \\ \hline
    & \vdots & \\ \hline
    \multirow{3}{*}{i} & $A_{i-1P}\, E_1$ & $x_{1+3i}$ \\
    & $A_{iP}$ & $x_{2+3i}$ \\
    & $A_{iP}\, E_2$ & $x_{3+3i}$ \\ \hline
    & \vdots & \\ \hline
    \multirow{3}{*}{n} & $A_{n-1P}\, E_1$ & $x_{1+3n}$ \\
    & $A_{nP}$ & $x_{2+3n}$ \\
    & $A_{nP}\, E_2$ & $x_{3+3n}$ \\ \hline
  \end{tabular}
  \caption{Assignment of variables to species}
  \label{tab:VarAssignment}
\end{table}

Assuming a distributive mechanism, a single phosphorylation occurs
with each encounter of substrate and kinase \cite{ptm-001} and $n$
phosphorylations therefore require $n$ encounters of substrate and
kinase of the form
\begin{displaymath}
  E_1 + A_{i-1P} \autorightleftharpoons{}{} A_{i-1 P}\, E_1
  \autorightarrow{}{} E_1 + A_{iP},\quad i=1, \ldots, n\ ,
\end{displaymath}
where $A_{0P} = A$ and $A_{1P} = A_P$. Similarly, $n$
dephosphorylations following a distributive mechanism require $n$
encounters of substrate and phosphatase of the form:
\begin{displaymath}
  E_2 + A_{iP} \autorightleftharpoons{}{} A_{i P}\, E_2
  \autorightarrow{}{} E_2 + A_{i-1P},\quad i=1, \ldots, n\ .
\end{displaymath}
Each phosphorylation and each dephosphorylation therefore consists of
3 reactions and consequently the network consists of $6n$
reactions. To each reaction we associate a rate constant. We use $k_i$
for phosphorylation  and $l_i$ for dephosphorylation reactions and 
obtain the following reaction network:
\begin{equation}
  \label{eq:network_dd}
  \begin{split}
    E_1 + A_{i-1P} \autorightleftharpoons{$k_{3i-2}$}{$k_{3i-1}$}
    A_{i-1 P}\, E_1 \autorightarrow{$k_{3i}$}{} E_1 + A_{iP},
    i=1,\, \ldots,\, n \\ 
    E_2 + A_{iP} \autorightleftharpoons{$l_{3i-2}$}{$l_{3i-1}$} A_{i
      P}\, E_2 \autorightarrow{$l_{3i}$}{} E_2 + A_{i-1P}, i=1,\,
    \ldots,\, n\ .
  \end{split}
\end{equation}
Using this notation, $k_{3i-2}$ ($l_{3i-2}$) denotes the association
constant, $k_{3i-1}$ ($l_{3i-1}$) the dissociation constant and
$k_{3i}$ ($l_{3i}$) the catalytic constant of the $i$-th
phosphorylation (dephosphorylation) step. We collect all rate
constants in a vector $\kappa$ 
defined in the following way:
\begin{definition}[The vector of rate constants $\kappa$]\label{def:rate_constants} \mbox{} \\
  We collect the six rate
  constants associated to the $i$-th phosphorylation/dephos\-phorylation in  network~(\ref{eq:network_dd})
  in a sub-vector  
  \begin{subequations}
    \begin{align}
      \label{eq:def_kappa_i}
      \kappa_{\brac{i}} &:= \brac{k_{3i-2},\, k_{3i-1},\,
        k_{3i},\,l_{3i-2},\, l_{3i-1},\, l_{3i}}^T\ . \\
      \intertext{We then combine the sub-vectors to a $6n$-dimensional column vector}
      \label{eq:def_kappa}
      \kappa &:= \col\brac{\kappa_{\brac{1}},\, \ldots,\,
        \kappa_{\brac{n}}}\ .
    \end{align}
  \end{subequations}
\end{definition}
We conclude this subsection with a brief comment on the notation used to
denote (sub-)vectors and matrices throughout this contribution.
\begin{remark}[Vector and matrix notation]\label{rem:notation}\mbox{}\rm \\
In the following sections the dimension of vectors is determined by
  the number of phosphorylation sites $n$ according to the formula
  $3n+3$. (The vector $x$, for example, lives in Euclidian
  $\R^{3n+3}$, cf. Table~\ref{tab:VarAssignment}). We will use the
  symbol $e_j$ to denote elements of the standard basis of Euclidian
  vector spaces and use the superscript $^{\brac{i}}$ to distinguish
  basis vectors of vector spaces of different dimension $3i+3$:
  \begin{align}
    \label{eq:def_en}
    \ei{i}_j &\text{\ldots denotes elements of the standard basis of
      $\R^{3i+3}$.} 
  \end{align}
  Likewise for matrices the superscript $^{\brac{i}}$ is used to indicate that
  the number of rows and/or columns depends on an integer $i$
  (see, for example, equations (\ref{eq:def_Zi}), (\ref{eq:S_direct}),
  (\ref{eq:Y_direct}) or (\ref{eq:Ei}) below).
  
  Later on we will split vectors of length $3n+3$ into
  consecutive sub-vectors of length three and we will use the
  subscript $_{\brac{i}}$ to denote the $i$-th sub-vector
    \begin{equation}      \notag 
      x_{\brac{i}} := \brac{x_{3i+1},\, x_{3i+2},\, x_{3i+3}}^T
      \end{equation}
 and we will use 
      $x = \col\brac{x_{\brac{0}},\, \ldots,\, x_{\brac{n}}}$
 to denote that $x$ consists of $n+1$ such
        sub-vectors. Likewise we will split vectors of length $6n$
        into sub-vectors
      \begin{equation}      \notag
      y_{\brac{i}} := \brac{y_{6i-5},\, \ldots,\, y_{6i}}^T 
      \end{equation}
of length six and use
      $y= \col\brac{y_{\brac{1}},\, \ldots,\, y_{\brac{n}}}$
 to denote that $y$ consists of $n$ such sub-vectors.\hbm
\end{remark}

\subsection{The dynamical system defined by the mass action
  network~(\ref{eq:network_dd})}
\label{sec:dyn_sys}

Starting with the stoichiometric matrix \Si{1} and the rate exponent
matrix \Yi{1} defined by the mass action network depicted in
Fig.~\ref{eq:net_phospho_1} one can recursively construct matrices
\Si{n} and \Yi{n} for the network (\ref{eq:network_dd}). 
Using the ordering of species and reactions introduced above in Table~\ref{tab:VarAssignment}
one obtains $\Si{n} \in \R^{\brac{3n+3} \times 6n}$, $\Yi{n} \in
\R^{\brac{3n+3} \times 6n}$ and $\Zi{n}$ by the following steps:

\vspace*{1mm}(I) Concerning $\Si{n}$:  \begin{subequations}
    \begin{align}
      \label{eq:def_S1}
      \Si{1} &= \left[
        \begin{array}{c}
          n_{11}\\ 
          n_{21}
        \end{array}
      \right],\quad 
      \Si{2} = \left[
        \begin{array}{c|c} 
          \Si{1} & 
          \begin{array}{c}
            n_{12}\\ 
            n_{22}
          \end{array} \\
          \hline
          {\textbf{0}}_{3\times 6} & n_{21}
        \end{array}
      \right],\quad \ldots \\[2mm] 
      \Si{j} &= \left[
        \begin{array}{c|c}
          \Si{j-1} &
          \begin{array}{c}
            n_{12}\\
            {\textbf{0}}_{3{(j-2)}\times 6}\\
            n_{22} 
          \end{array}\\
          \hline
          {\textbf{0}}_{3\times 6(j-1)} & n_{21}
        \end{array}\right]
      \label{eq:def_N_n}
    \end{align}
  \end{subequations}
  for $j=3,...,n$ with the following sub-matrices of dimension $3\times 6$:
  \begin{subequations}\label{eq:nij}
    \begin{align}
      n_{11} &= \left[
        \begin{array}{rrrrrr}
          -1& \phantom{-}1 & \phantom{-}1 & \phantom{-}0 & \phantom{-}0 &
          \phantom{-}0 \\ 
          -1& 1 & 0 & 0 & 0 & 1 \\
          0 & 0 & 0 & -1& 1 & 1 
        \end{array}
      \right],\\
      n_{12} &= \left[
        \begin{array}{rrrrrr}
          -1 & \phantom{-}1 & \phantom{-}1 & \phantom{-}0 & \phantom{-}0 &
          \phantom{-}0   \\ 
          0 & 0 & 0 & 0 & 0 & 0   \\
          0 & 0 & 0 & -1& 1 & 1 
        \end{array}
      \right],\\
      n_{21} &= \left[
        \begin{array}{rrrrrr} 
          \phantom{-}1 & -1& -1& \phantom{-}0 & \phantom{-}0 & \phantom{-}0 \\
          0 & 0 & 1 & -1& 1 & 0 \\
          0 & 0 & 0 & 1 & -1& -1 
        \end{array}
      \right],\\
      n_{22} &= \left[
        \begin{array}{rrrrrr}
          \phantom{-}0 & \phantom{-}0 & \phantom{-}0 & \phantom{-}0 &
          \phantom{-}0 & \phantom{-}0 \\ 
          -1& 1 & 0 & 0 & 0 & 1 \\
          0 & 0 & 0 & 0 & 0 & 0 
        \end{array}
      \right].
    \end{align}
For fixed $n$, the recursive formula (\ref{eq:def_N_n}) evaluates to \\[1mm] 
\begin{align}
      \label{eq:S_direct}
      \Si{n} &:= \left[ 
        \begin{array}{c|c|c|c|c|c}
          n_{11} & n_{12} & n_{12} & n_{12} & & n_{12} \\
          n_{21} & n_{22} & 0_{3\times 6} & \multirow{2}{*}{$0_{2\cdot
              3\times 6}$} & &
          \multirow{5}{*}{$0_{\brac{n-2}\cdot 3\times 6}$} \\
          \cline{1-1}
          0_{3\times 6} & n_{21} & n_{22} & \phantom{0_{3\times 6}} &
          \dots & \\ \cline{1-2}
          \multicolumn{2}{c|}{0_{3\times 6\cdot 2}} & n_{21} & n_{22}
          & & \\ \cline{1-3}
          \multicolumn{3}{c|}{0_{3\times 6\cdot 3}} & n_{21} & & \\ 
          \cline{1-4}
          \multicolumn{4}{c|}{\vdots} &\ddots & \\ \cline{1-5}
          \multicolumn{5}{c|}{0_{3\times 6\cdot \brac{n-1}}} & n_{21}
        \end{array}
      \right]\, \in \R^{(3n+3)\times 6n}\, . 
\end{align}
  \end{subequations}


\vspace*{2mm}(II)  Concerning $\Yi{n}$: \begin{subequations}
    \begin{align}
      \Yi{1} &= \left[
        \begin{array}{cccccc}
          \ei{1}_1+\ei{1}_2 & \ei{1}_4 & \ei{1}_4 & \ei{1}_3+\ei{1}_5
          &\ei{1}_6 & \ei{1}_6
        \end{array}
      \right]\, ,
      \label{eq:def_Y_1} \\
      \Yi{j} &= \left[
        \begin{array}{c|c}
          \begin{array}{c} 
            \Yi{j-1} \\ 
            {\bf{0}}_{3\times 6(j-1)}
          \end{array} & 
          \begin{array}{cccccc} 
            \ei{j}_1 + \ei{j}_{3j-1} & \ei{j}_{3j+1} & \ei{j}_{3j+1} &
            \ei{j}_3 + \ei{j}_{3j+2} & \ei{j}_{3j+3} & \ei{j}_{3j+3}
          \end{array}
        \end{array}
      \right]
      \label{eq:def_Y_n}
    \end{align}
for $j=2,...,n$ with $\Yi{1} \in \R^{6\times 6}$ (cf. eq.(\ref{eq:1ys_ode}))
and $\Yi{j}\in \R^{(3j+3)\times
6j}$. 
Using\begin{align}
      \Yo{i} &:=\left[
        \begin{array}{cccccc}
          \e{i}_1 + \e{i}_{3i-1} & \e{i}_{3i+1} & \e{i}_{3i+1} &
          \e{i}_3 + \e{i}_{3i+2} & \e{i}_{3i+3} & \e{i}_{3i+3}
        \end{array}
      \right] \, \in \R^{(3i+3)\times 6}\, ,\end{align}
      the recursive formula (\ref{eq:def_Y_n}) evaluates to\\[1mm]
      \begin{align}\label{eq:Y_direct}
      {\Yi{n}}^T &:= 
      \left[
        \begin{array}{c|c|c|c|c|c|c}
          \Yo{1}^T & 0_{6\cdot 1 \times3} & \multirow{2}*{$0_{6\cdot
              2 \times 3}$} & \multirow{3}*{$0_{6\cdot 3 \times 3} $}&
          \phantom{0_{6\cdot 2 \times 3}} & \phantom{0_{6\cdot 2
              \times 3}} &  \multirow{5}*{$0_{6\cdot n \times 3}$} \\
          \cline{1-2}
          \multicolumn{2}{c|}{\Yo{2}^T} & & & & & \\ \cline{1-3}
          \multicolumn{3}{c|}{\Yo{3}^T} & & & \\
          \cline{1-4} \multicolumn{5}{r|}{\ddots} &  \\ \cline{1-5}
          \multicolumn{6}{c|}{\Yo{n-1}^T} &  \\ \hline
          \multicolumn{7}{c}{\Yo{n}^T}\\
        \end{array}
      \right]\, .
    \end{align}
\end{subequations}


\vspace*{2mm}(III)  Concerning $\Zi{n}$:\\[1mm] 
    \begin{equation}
    \label{eq:def_Zi}
    \Zi{n} = 
    \begin{array}{cc}
      \left[
        \begin{array}{ccc|}
          1 & 0 & 0 \\
          0 & 0 & 1 \\
          0 & 1 & 0
        \end{array}\,
      \right.
      &
      \!\underbrace{
        \!\left.
          \begin{array}{ccc|c|ccc}
            1 & 0 & 0 & & 1 & 0 & 0 \\
            0 & 0 & 1 & \cdots & 0 & 0 & 1 \\
            1 & 1 & 1 & & 1 & 1 & 1
          \end{array}
        \right]
      }_{n-times}\, \in \R^{3\times (3n+3)}.
    \end{array}
  \end{equation}
We note that the three rows of $\Zi{n}$ form a basis for the left kernel of $\Si{n}$.

\vspace*{5mm}
\begin{definition}[Reaction rate vector defined by
  network~(\ref{eq:network_dd})]\label{def:dyn_sys_n_phos} \mbox{}\\  
  The vector of rate constants $\kappa$ from (\ref{eq:def_kappa}) and
  the columns $y_i$ of \Yi{n} from (\ref{eq:Y_direct})
  define 
  two monomial functions $ \Phi^{\brac{n}}:\R^{3n+3} \to
  \R^{6n}$ and $\ri{n}{\kappa}{x}:\R^{3n+3}\to\R^{6n}$:
  \begin{subequations}
    \begin{align}
      \label{eq:def_phinx}
      \Phi^{\brac{n}}\brac{x} &:= \brac{x^{y_1},\, \ldots,\,
        x^{y_{6n}}}^T \\ 
      \intertext{and}
      \label{eq:def_rnkx}
      \ri{n}{\kappa}{x} &:= \diag\brac{\kappa}\,
      \Phi^{\brac{n}}\brac{x} 
    \end{align}
  \end{subequations}
  for $x\in \Rnn^{3n+3}$
  (cf. (\ref{eq:def_xy})). The $6n$-dimensional vector
  $\ri{n}{\kappa}{x}$ is called the reaction rate vector.
\end{definition}
Together with the stoichiometric matrix
\Si{n} from (\ref{eq:S_direct}), the dynamical system,
defined by the network~(\ref{eq:network_dd}) of Fig.~\ref{fig-1}, is
then given by
\begin{equation}
  \label{eq:ode_def}
  \dot x = \Si{n}\, \ri{n}{\kappa}{x}
\end{equation}

\begin{remark}[The monomial function
  $\Phi^{\brac{n}}\brac{x}$]\label{rem:rkx_Phi_x} \mbox{}\rm\\ 
  Observe that the monomial
  function  $\Phi^{\brac{n}}\brac{x}$ in (\ref{eq:def_rnkx}) satisfies 
  \begin{subequations}
    \begin{align}\label{eq:YT_mu}
      \Phi^{\brac{n}}\brac{e^\mu} &= \ e^{{\mathcal{Y}^{(n)}}^T\, \mu}\, , \\
      \label{eq:1_over_phi}
      \frac{1}{\Phi^{\brac{n}}\brac{x}} &= \ \Phi^{\brac{n}}\brac{x^{-1}}
    \end{align}
  \end{subequations}
  for vectors $\mu\in\R^{3n+3}$ and $x\in\R^{3n+3}$ with $x_i\neq 0$, 
  $i=1$, \ldots, $3n+3$. \hbm 
\end{remark}

From equation~(\ref{eq:ode_def}) follows that level sets $\{x\in\R^{3n+3} |
\Zi{n}\, x = {const.}\}$ are invariant under the flow of (\ref{eq:ode_def})
as $\Zi{n}\, x(t) = \Zi{n}\, x(0)$ along solutions $x(t)$ of
(\ref{eq:ode_def}).
This observation motivates
the classical definition of multistationarity 
originating in 
chemical engineering (cf.\ Figure \ref{fig:cartoon_multi} and, for
example, 
\cite[Definition~1\&Remark~2]{fein-043} or \cite{fm-95a,fm-95b}):
\begin{definition}[Multistationarity]\label{defimulti} \mbox{}\\
  The system $\dot x = \Si{n}\, \ri{n}{\kappa}{x}$ from \eqref{eq:ode_def} is
  said to exhibit multistationarity if and only if there exist a
  positive vector $\kappa \in \Rp^{6n}$ and at least two distinct positive
  vectors $a$, $b \in \Rp^{3n+3}$ with
  \begin{subequations}
    \begin{align}
      \label{eq:multistat_ode_x0}
      \Si{n}\, \ri{n}{\kappa}{a} &= 0 \\
      \label{eq:multistat_ode_x1}
      \Si{n}\, \ri{n}{\kappa}{b} &= 0 \\
      \label{eq:multistat_con_rel_x0_x1}
      \Zi{n}\, a &= \Zi{n}\, b.
    \end{align}
  \end{subequations}
\end{definition}
\begin{figure}
  \centering
  \begin{minipage}[c]{.4\linewidth}
    \includegraphics[width=1\linewidth,angle=-90]{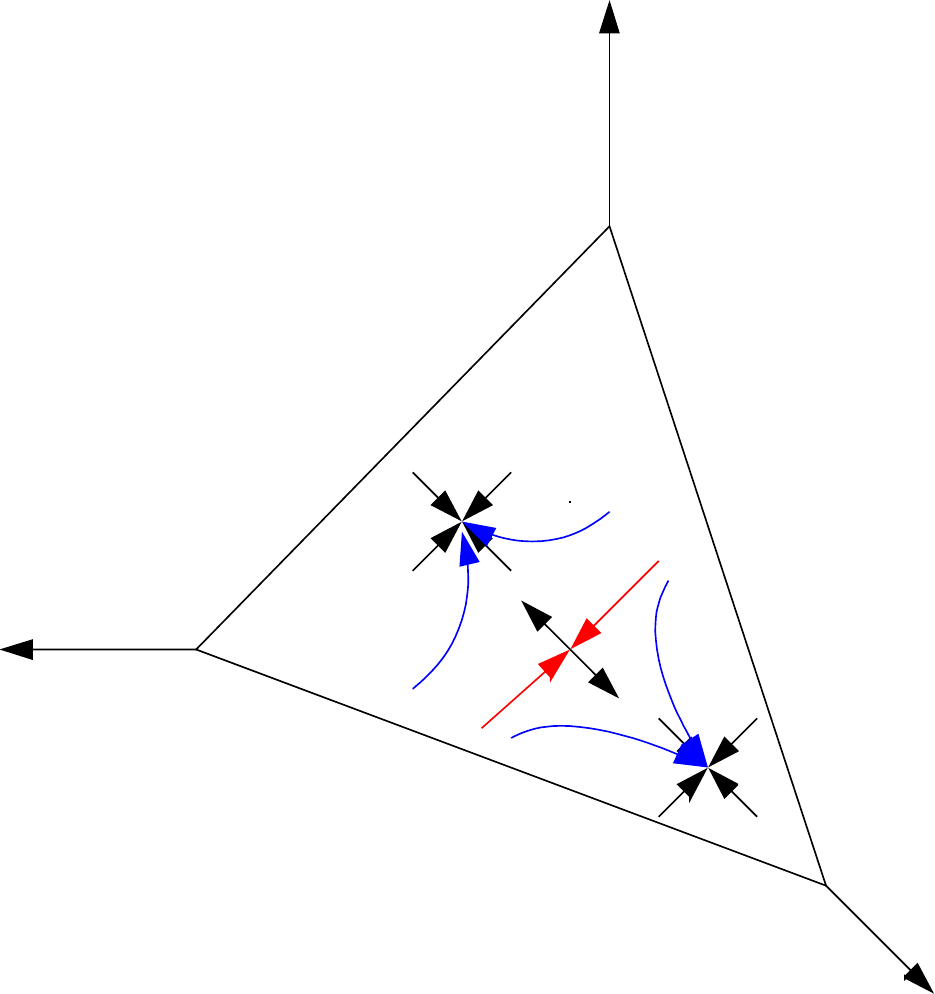}
  \end{minipage}
  \caption{Multistationarity and bistability:
    While the equation $\Si{n}\, r^{(n)}(\kappa,x)=0$ often has
    infinitely many steady state solutions, a solution $x(t)$ of the
    differential equation $\dot x = \Si{n}\, r^{(n)}(\kappa ,x)$ \lq
    sees\rq{} only those contained in the level set $L_{x(0)}:=\{x \in
    \Rnn^{3n+3} | \Zi{n}\, x = \Zi{n}\, x(0) \}$. 
    Multistationarity
    requires at least two positive steady state solutions in
    $L_{x(0)}$, bistability two \lq stable\rq{} positive steady state
    solutions in $L_{x(0)}$. 
    The triangle in the figure shows such a level
    set $L_{x(0)}$ in $\Rnn^3$ with two stable positive steady states. The stable
    manifold (in red) of the third positive steady state generates the threshold
    separating the domains of attraction of the two stable steady
    states.
    \label{fig:cartoon_multi}
  }
\end{figure}

We conclude this section by a discussion of $\ker\brac{\Si{n}}$ and 
of the pointed polyhedral cone $\ker\brac{\Si{n}} \cap \Rnn^{3n+3}$:
\begin{lem}[A nonnegative basis of $\ker\brac{\Si{n}}$]\label{lem:basis_ker_S} \mbox{} \\
  Let $n\geq 2$ be given and recall the stoichiometric matrix \Si{n}
  as defined in (\ref{eq:S_direct}). Let
  \begin{subequations}
    \begin{align}
      \label{eq:E1}
      E &= \left[
        \begin{array}{ccc}
          1&0&1\\
          1&0&0\\
          0&0&1\\ 
          0&1&1\\
          0&1&0\\
          0&0&1
        \end{array}
      \right] \\
      \intertext{and define the $6n\times 3n$ matrix with $E$ on the
        block diagonal (and zero otherwise):}
      \label{eq:Ei}
      \Ei{n} &:= \left[
        \begin{array}{ccc}
          E & & \\
          & \ddots & \\
          & & E
        \end{array}
      \right]
    \end{align}
    Then columns of \Ei{n} form a basis of $\ker\brac{\Si{n}}$. In addition,
    the columns of \Ei{n} are generators of $\ker\brac{\Si{n}} \cap
    \Rnn^{3n+3}$.
  \end{subequations}
\end{lem}
\begin{proof}
  First recall the building blocks of \Si{n}, the $3 \times 6$
  sub-matrices $n_{11}$, $n_{21}$, $n_{12}$ and $n_{22}$ (cf.~(\ref{eq:nij})). We show that
  $\im\brac{\Ei{n}} = \ker\brac{\Si{n}}$ by proving the two inclusions
  $\im\brac{\Ei{n}} \subseteq \ker\brac{\Si{n}}$ 
  and $\ker\brac{\Si{n}} \subseteq \im\brac{\Ei{n}}$.
  
   \vspace{1mm}(i)\ Concerning $\im\brac{\Ei{n}} \subseteq
      \ker\brac{\Si{n}}$:\\
    Straightforward computations show that $n_{11}\, E = 0_{3\times
      3}$, $n_{12}\, E = 0_{3\times 3}$, $n_{21}\, E = 0_{3\times 3}$
    and $n_{22}\, E = 0_{3\times 3}$. Thus $\Si{n}\, \Ei{n} =
    0_{(3n+3)\times 3n}$.

  \vspace{1mm}(ii)\ Concerning $\ker\brac{\Si{n}} \subseteq
  \im\brac{\Ei{n}}$: \\
  We first note
  that the columns of $E$ form a basis of $\ker\brac{n_{11}}$ and
  $\ker\brac{n_{21}}$.
  Pick any vector $\eta\in\ker\brac{\Si{n}}$
  and split it in consecutive sub-vectors of length 6 as described
  in Remark~\ref{rem:notation}: 
  $\eta=\col\brac{\eta_{\brac{1}},\,
    \ldots,\,\eta_{\brac{n}}}$.  The vector $\eta$ satisfies
  $\Si{n}\, \eta = 0$, that is (cf. eq.(\ref{eq:S_direct}) defining
  \Si{n})
  \begin{align*}
    n_{11}\eta_{(1)} + \sum\limits_{i=2}^n n_{12}\eta_{(i)} & = 0,
    \\ 
    n_{21}\eta_{(1)} + n_{22}\eta_{(2)} & = 0, \\
    {} & \vdots \\
    n_{21}\eta_{(n-1)} + n_{22}\eta_{(n)} & = 0,\\
    n_{21}\eta_{(n)} & = 0.
  \end{align*}
  From the last equation follows that $\eta_{(n)} \in \ker(n_{21}) =
  \im\brac{E}$.  The preceding equations then imply
  $\eta_{(i)}\in \ker(n_{21}) = \im\brac{E}$, for $i=n-1$, \ldots,
  $1$. Hence $\eta\in\im\brac{\Ei{n}}$ and thus $\ker\brac{\Si{n}}
  \subseteq \im\brac{\Ei{n}}$.
  
  \vspace{1mm} Finally, concerning the pointed polyhedral cone, observe that a vector
  $\eta$ is called a generator of the cone
  $\ker\brac{\Si{n}}\cap\Rnn^{3n+3}$, if and only if $\eta$ satisfies
  the following three conditions (cf.\ for example, \cite{kst-04}): 
  \begin{eqnarray*}
    (i)\quad \eta  \geq 0\, , \qquad  (ii) \quad    \Si{n}\, \eta  =
    0,\phantom{,,,,,,,,,,,,,,}\\
    (iii)\ \ \mbox{for  generators $\eta_1$ and $\eta_2$ of
      $\ker\brac{\Si{n}} \cap \Rnn^{3n+3}$ one has}\phantom{....}\\
    \supp\brac{\eta_1} \subseteq \supp\brac{\eta_2}
    \Rightarrow \eta_1 = 0\; \text{or $\eta_1 =
      \alpha\eta_2$ with $\alpha >0$,\phantom{..}}
  \end{eqnarray*}
  where  the symbol
  $\supp\brac{\eta}$ denotes the  support of a vector $\eta\in\R^{3n+3}$
  (i.e.\ the set of indices
  of the nonzero elements). Observe that the column vectors of $E$
  satisfy these conditions and recall the block diagonal structure of
  \Ei{n}. Hence the columns of \Ei{n} satisfy these conditions as
  well.
\end{proof}

\vspace*{5mm}\section{Results}
\label{sec:results}

In this section we first introduce a condition for multistationarity
(cf.\ Definition~\ref{defimulti})  in terms of sign patterns, followed
by a discussion of sign patterns satisfying this condition.
The rather technical proofs
of Theorem~\ref{theo:poly_sol} and Theorem~\ref{theo:feasible_patterns} 
will be presented in  Section~\ref{sec:proof_poly_sol} and
Section~\ref{sec:proof_theo_patterns} respectively.

\subsection{A sign condition for multistationarity in the dynamical
  system~(\ref{eq:ode_def})}
\label{sec:sign_condi_multi}

We now turn to multistationarity for dynamical systems defined by
network (\ref{eq:network_dd}), that is to the system defined in
(\ref{eq:ode_def}). It follows from Definition~\ref{defimulti} that
multistationarity requires, for a given vector $\kappa$, the existence
of two positive solutions $a$ and $b$ to the polynomial equations
(\ref{eq:multistat_ode_x0}) and (\ref{eq:multistat_ode_x1}) that
satisfy the linear condition (\ref{eq:multistat_con_rel_x0_x1}). In
Theorem~\ref{theo:poly_sol} below, we first turn to the polynomial
equations. The linear constraint is taken into account afterwards in
Corollary~\ref{coro:sign_condi}.  In the following
Theorem~\ref{theo:poly_sol} we use matrices \PPi{n} and \Mi{n} to
argue that, for system~(\ref{eq:ode_def}), the existence of {\em
  positive solutions} to the polynomial equations is equivalent to the
existence of {\em solutions of a linear system}.
\begin{definition}[Matrices \PPi{n} and \Mi{n}]\label{def:Pi_and_Mi}\mbox{} \\
  Let $\uu{1}$ denote a vector of dimension 6 filled with
  the number 1. Define the $6\times 2$ matrix
  \begin{subequations}
    \begin{align}
      \label{eq:def_Pi_0}
      \Pi_0\brac{i} :&= \left[
        \begin{array}{cc}
          \brac{2-i}\, \uu{1} & \brac{i-1}\, \uu{1}
        \end{array}
      \right],\; i=1,\, \ldots,\, n \\
      \intertext{and the $6\, n\times 2$ matrix}
      \label{eq:def_Pi}
      \PPi{n} &:= \left[
        \begin{array}{c}
          \Pi_0\brac{1} \\ \vdots\\ \Pi_0\brac{n}
        \end{array}
      \right]\ .
    \end{align}
  \end{subequations}
  Further define the $3\times 3$ matrices 
  \begin{subequations}
    \begin{align}
      \label{eq:def_M0}
      M\brac{0,n} &:=
      \begin{bmatrix}
        -1 & -n+1 & \phantom{-}n\\
        \phantom{-}1 & \phantom{-}n & -n\\
        -1 & -n+2 & \phantom{-}n-1
      \end{bmatrix} \\
      \intertext{and}
      \label{eq:def_Mi}
      M\brac{i,n} &:=
      \begin{bmatrix}
        0 & -i+2 & \phantom{-}i-1 \\
        1 & \phantom{-}n-i & -n+i\\
        0 & -i+2 & \phantom{-}i-1
      \end{bmatrix} \\
      \intertext{for $i=1,...,n$  and the $\brac{3\, n+3}\times 3$ matrix}
      \label{eq:def_Mn}
      \Mi{n} &:= \left[
        \begin{array}{c}
          M\brac{0,n} \\
          M\brac{1,n} \\
          \vdots \\
          M\brac{n,n}
        \end{array}
      \right]\, .
    \end{align}
  \end{subequations}
\end{definition}
In Section~\ref{sec:proof_poly_sol} below we prove the following result.
\begin{theorem}[Solutions to the polynomial equations]\label{theo:poly_sol}\mbox{} \\
  Consider the $n$-site sequential distributed phosphorylation
  network~(\ref{eq:network_dd}) and recall the corresponding matrices
  \Si{n}, \Yi{n},  \Ei{n},  \PPi{n} and \Mi{n} 
  (cf. eqns.~(\ref{eq:S_direct}), (\ref{eq:Y_direct}), (\ref{eq:Ei}),
  (\ref{eq:def_Pi})  and (\ref{eq:def_Mn})). The following are
  equivalent:
  \begin{enumerate}[{(}A{)}]
  \item There exists a vector $\kappa\in\Rp^{6n}$ and vectors $a$,
    $b\in\Rp^{3n+3}$ with $a\neq b$ satisfying
    (\ref{eq:multistat_ode_x0}) \& (\ref{eq:multistat_ode_x1}), that
    is
    \begin{displaymath}
      \Si{n}\, \ri{n}{\kappa}{a} = 0,\; \Si{n}\, \ri{n}{\kappa}{b} = 0.
    \end{displaymath}
  \item There exist  vectors $\mu\in\R^{3n+3}$, $\mu\neq 0$, and 
     $\brac{\nu,\lambda}\in\Rp^{3n}\times\Rp^{3n}$,
      such that
    \begin{equation}
      \label{eq:trans_eq}
      {\Yi{n}}^T\, \mu = \ln\frac{\Ei{n}\, \nu}{\Ei{n}\, \lambda}\ .
    \end{equation}
  \item There exist  vectors $\mu\in\R^{3n+3}$, $\xi\in\R^2$ 
  and $\nu$, $\lambda \in \Rp^{6n}\times \Rp^{6n}$, with  $\mu\neq 0$
  such that
    \begin{subequations}
      \begin{enumerate}
      \item the vectors $\mu$ and $\xi$ satisfy
        \begin{equation}
          \label{eq:mu_lin}
          {\Yi{n}}^T\, \mu = \PPi{n}\, \xi,
        \end{equation}
      \item and the vectors $\nu$, $\lambda$ and $\xi$ satisfy
        \begin{align}
          \label{eq:lambda_free}
          \lambda &\in\Rp^{6n},\quad \text{free,} \\
          \intertext{and, for $i=1$, $\ldots$, $n$}
          \label{eq:def_nu_1}
          \nu_{3i} &= \lambda_{3i}\,  e^{\brac{2-i}\, \xi_1+\brac{i-1}\,
            \xi_2},  \\ 
          \nu_{3i-2} &= \lambda_{3i-2}\, \frac{\nu_{3i}}{\lambda_{3i}},
          \\ 
          \label{eq:def_nu_3}
          \nu_{3i-1} &= \lambda_{3i-1}\, \frac{\nu_{3i}}{\lambda_{3i}}.
        \end{align}
      \end{enumerate}
    \end{subequations}
  \item There exists a vector $\mu\in\R^{3n+3}$, $\mu\neq 0$ with
    \begin{equation}
      \label{eq:mu_im_M}
      \mu\in\im\brac{\Mi{n}}\, .
    \end{equation}
  \end{enumerate}
\end{theorem}

\begin{remark}[Solutions $a$, $b$ and $\kappa$ to the polynomial
  equations~(\ref{eq:multistat_ode_x0})\&(\ref{eq:multistat_ode_x1})]\label{rem:solutions_polynomials}\mbox{}
  \rm \\ 
 Recall the matrices \Mi{n} and \Ei{n} (cf.\ eq.~(\ref{eq:def_Mn}) and
   eq.~(\ref{eq:Ei})) and  the
  monomial function $\Phi^{\brac{n}}$ from eq.(\ref{eq:def_rnkx})
  with the properties (\ref{eq:YT_mu}) and (\ref{eq:1_over_phi}) stated in
  Remark~\ref{rem:rkx_Phi_x}. Fix any two
  vectors $\mu\in\im\brac{\Mi{n}}$ and $\lambda\in\Rp^{6n}$. Then
  vectors ($\kappa$, $a$) and ($\kappa$, $b$) satisfying the
  polynomial equations (\ref{eq:multistat_ode_x0}) \&
  (\ref{eq:multistat_ode_x1}) are determined by 
  \begin{displaymath}
    \diag(\kappa)\Phi^{(n)}(e^{\ln a}) = E^{(n)}\nu\, , \quad
    \diag(\kappa)\Phi^{(n)}(e^{\ln b})=E^{(n)}\lambda\, , \quad
    \brac{\nu,\lambda}\in\Rp^{3n}\times\Rp^{3n}\, .
  \end{displaymath}
  By (\ref{eq:YT_mu}), they fulfill (\ref{eq:trans_eq}) for
  $\mu=\ln\brac{\frac{a}{b}}$. Thus, vectors ($\kappa$, $a$) and
  ($\kappa$, $b$) given by
  \begin{subequations}
    \begin{align}
      a &\text{, free in $\Rp^{3n+3}$,} \\\label{eq0:ln_b-ln_a_im_M}
      b &:= \diag\brac{e^\mu}\, a,\\
      \label{eq:def_k}
      \kappa &:= \diag\brac{\Phi^{\brac{n}}\brac{a^{-1}}}\, \Ei{n}\, \lambda.
    \end{align}
  \end{subequations}
  are solutions to the polynomial
  equations~(\ref{eq:multistat_ode_x0})\&(\ref{eq:multistat_ode_x1}).
  Note that, in particular, the vector $a\in\Rp^{3n+3}$ is free and
  that $a$ and $b$ satisfy by construction
  \begin{equation}
    \label{eq:ln_b-ln_a_im_M}
    \ln b - \ln a = \mu \in\im\brac{\Mi{n}}\, .
  \end{equation}
  This follows from the proof of Theorem~\ref{theo:poly_sol}, given in
  Section~\ref{sec:proof_poly_sol} below. See also
  \cite[Remark~7]{fein-043}. 
  In Section~\ref{sec:discussion}, we will interpret the fact that
  eq.~(\ref{eq:def_Mi}) implies $\mu_{3i+1}=\mu_{3i+3}$ in
  eq.~(\ref{eq:ln_b-ln_a_im_M}) for $i=1,2,,...,n$.
  \hbm
\end{remark}

As a consequence of Theorem~\ref{theo:poly_sol}, any two distinct solutions
$a$ and $b$ to the polynomial equations (\ref{eq:multistat_ode_x0}) \&
(\ref{eq:multistat_ode_x1}) satisfy the
condition~(\ref{eq:ln_b-ln_a_im_M}) with $\mu\neq 0$. And vice versa
in case $\kappa$ is chosen according to (\ref{eq:def_k}).
From Definition~\ref{defimulti}  
follows that for multistationarity $a$ and $b$ {\em additionally} have
to satisfy ${\Zi{n}}\, \brac{b-a} = 0$ and hence
$b-a\in\im\brac{\Si{n}}$. The existence of vectors $a$ and $b$
satisfying both conditions is covered by
Corollary~\ref{coro:sign_condi} given below (this result itself is a
consequence of \cite[Lemma~1]{cc-08}; see also \cite{fein-043} and,
for example, \cite{fm-95b} for an earlier, more informal
discussion). The corollary is based on the sign patterns of vectors
and sets thereof.
\begin{definition}[Sign patterns defined by vectors and linear subspaces]\label{signsub} \mbox{} \\
  Given $v\in\R^{m}$, a vector $\sigma\in\{-1,0,1\}^{m}$ with
  \begin{subequations}
    \begin{equation}
      \label{eq:def_sign_pattern}
      \sigma_i = \sig{v_i}
    \end{equation}
    is called the {\rm sign pattern} of $v$: $\sigma=\sig{v}$. Given a
    linear subspace  $V\subseteq \R^{m}$, 
    the set $\sig{V}$ of sign patterns  is defined by
    \begin{equation}
      \label{eq:def_sign_pattern_set}
      \sig{V} := \left\{ \sigma\in\{-1,0,1\}^{m}\; |\; \exists v\in
        V\quad \text{with}\quad \sigma = \sig{v} \right\}.
    \end{equation}
  \end{subequations}
\end{definition}
Now we can state the announced result:
\begin{coro}[cf. Lemma~1 of \cite{cc-08}]\label{coro:sign_condi} \mbox{} \\
  Two vectors $a$, $b\in\Rp^{3n+3}$ with
  \begin{subequations}\label{eq:ab26}
    \begin{align}
      \label{eq:a_neq_b}
      a &\neq b \\
      \label{eq:ln_condi}
      \ln b - \ln a &\in \im\brac{\Mi{n}}\\
      \label{eq:diff_condi}
      b-a &\in \im\brac{\Si{n}}
    \end{align}
  \end{subequations}
  exist if and only if
  \begin{equation}
    \label{eq:sign_condi}
    \matMi{n} := \sig{\im\brac{\Mi{n}}} \cap
    \sig{\im\brac{\Si{n}}} \neq \{0\}\ . 
  \end{equation}
  In case (\ref{eq:sign_condi}) holds, every pair $\mu\in\im\brac{\Mi{n}}$,
  $s\in\im\brac{\Si{n}}$ with
  \begin{equation}
    \sig{\mu} = \sig{s} \neq 0
  \end{equation}
  defines a pair of positive vectors $a$ and $b$ with the property \eqref{eq:ab26} via: 
  \begin{subequations}
    \begin{align}
      \label{eq:def_a_1}
      a &= \brac{a_i}_{i=1,\ldots,3n+3} \\
      \intertext{with}
      \label{eq:def_a_2}
      a_i &=
      \begin{cases}
        \frac{s_i}{e^{\mu_i} -1},\; \text{if
          $\mu_i\neq 0$} \\
        \bar a_i>0\text{, free,}
      \end{cases} \\
      \label{eq:def_b}
      b &= \diag\brac{e^{\mu}}\, a.
    \end{align}
  \end{subequations}
\end{coro}
\begin{proof}
  Apply \cite[Lemma~1]{cc-08} with $a$ as $p$, $b$ as $q$,
  $\im\brac{M^{(n)}}$ as $M_1$ and $\im\brac{\Si{n}}$ as $M_2$.
\end{proof}
As a consequence one can state a condition for multistationarity in
system (\ref{eq:ode_def}) defined by the network
(\ref{eq:network_dd}) based on the sign patterns defined by the linear
subspaces $\im\brac{\Si{n}}$ and $\im\brac{\Mi{n}}$ and the set
\matMi{n}. We will argue in Remark~\ref{rem:matro-condi} that this
condition is satisfied if and only if at least one out of
$\frac{1}{2}\, \big(3^{3n+3}-1\big)$ linear systems is feasible. Hence
one may in fact establish multistationarity for the polynomial
system~(\ref{eq:ode_def}) by an analysis of linear inequality
systems.
\begin{theorem}[A sign condition for multistationarity of
  systems~(\ref{eq:ode_def})]\label{theo:multi_linear_problem} \mbox{} \\
  Consider the $n$-site sequential distributed phosphorylation
  network~(\ref{eq:network_dd}) and recall the corresponding matrices
  \Si{n},   \Zi{n} and  \Mi{n},
  (cf. eqns.~(\ref{eq:S_direct}),  (\ref{eq:def_Zi})) and (\ref{eq:def_Mn}).
  There exists
  vectors $a$, $b\in\Rp^{3n+3}$ and $\kappa\in\Rp^{6n}$ satisfying
  \begin{displaymath}
    \Si{n}\, \ri{n}{\kappa}{a}=0,\; \Si{n}\, \ri{n}{\kappa}{b}=0,\; {\Zi{n}}\,
    \brac{b-a} = 0,
  \end{displaymath}
  if and only if
  \begin{displaymath}
    \matMi{n} =\sig{\im\brac{\Mi{n}}} \cap
    \sig{\im\brac{\Si{n}}}\neq \{0\}.
  \end{displaymath}
\end{theorem}
\begin{proof}
  Suppose that $\matMi{n}\neq \{0\}$.
  Then, by Corollary~\ref{coro:sign_condi} vectors $a$ and $b$
  satisfying (\ref{eq:a_neq_b}) -- (\ref{eq:diff_condi}) are given by
  (\ref{eq:def_a_1}) -- (\ref{eq:def_b}). By
  Remark~\ref{rem:solutions_polynomials} (where $a>0$ is free and may
  be chosen as in (\ref{eq:def_a_1}) \& (\ref{eq:def_a_2})) and
  Theorem~\ref{theo:poly_sol} there exists a vector
  $\kappa\in\Rp^{6n}$ given by eq.(\ref{eq:def_k}) such that $a$, $b$
  and $\kappa$ are solutions of the polynomials
  (\ref{eq:multistat_ode_x0}) \& (\ref{eq:multistat_ode_x1}). By
  (\ref{eq:diff_condi}) $a$ and $b$ additionally satisfy $\Zi{n}\,
  \brac{b-a} = 0$.

  Vice versa, let $a$, $b\in\Rp^{3n+3}$ with $a \neq b$ and
  $\kappa\in\Rp^{6n}$ be given, such that (\ref{eq:multistat_ode_x0}),
  (\ref{eq:multistat_ode_x1}) \& 
  (\ref{eq:multistat_con_rel_x0_x1}) hold. From
  eq.~(\ref{eq:mu_im_M}) of Theorem~\ref{theo:poly_sol} follows
  that $\mu := \ln b -\ln a \in \im\brac{\Mi{n}}$ (cf.\ also
  Remark~\ref{rem:solutions_polynomials},
  eq.~(\ref{eq:ln_b-ln_a_im_M})) and, as $\Zi{n}\, \brac{b-a}=0
  \Leftrightarrow b-a \in \im\brac{\Si{n}}$,
  Corollary~\ref{coro:sign_condi} implies that (\ref{eq:sign_condi})
  holds.
\end{proof}
\begin{remark}[The condition $\matMi{n}\neq \{0\}$]\label{rem:matro-condi} \mbox{} \rm  
  \begin{enumerate}[{(}A{)}]
  \item $\matMi{n}\neq \{0\}$ is equivalent to feasibility
      of (at least) one linear system:\
    We  emphasize the relation between condition
    (\ref{eq:sign_condi}) and linear inequality systems.
    An element $\sigma\in\{-1,0,1\}^{3n+3}$ is an element of $\matMi{n}$ from
    (\ref{eq:sign_condi}) if and only if the following linear system
    is feasible:
    \begin{equation}
      \label{eq:lin}
      \begin{split}
        &\exists s\in\R^{3n+3}\;  \text{and $\xi\in\R^3$ such that} \\
        &{\Zi{n}}\, s = 0\text{, with $\sigma_i\, s_i >0$ if
          $\sigma_i\neq 0$ and $s_i = 0$ else, \ and} \\
        &\mu = \Mi{n}\, \xi\text{, with $\sigma_i\, \mu_i >0$ if
          $\sigma_i\neq 0$ and $\mu_i = 0$ else.}
      \end{split}
    \end{equation}
  \item Elements of $\{-1,1\}^{3n+3}$: \
    Recall the formulae~(\ref{eq:def_a_1})--(\ref{eq:def_b}) and
    observe that $\mu_i = 0 $ $\Leftrightarrow$ $a_i = b_i$. Hence a
    vector $\mu$ (and corresponding $s$) with $\sig{\mu} \in
    \{-1,1\}^{3n+3}$ yields $a$ and $b$ that differ in every
    coordinate, while a  vector $\mu$ (and corresponding $s$) with
    $\sig{\mu} \in \{-1,0,1\}^{3n+3}$, $\mu\neq 0$ yields $a$ and $b$
    that may be equal in some coordinates.
  \item Testing $\matMi{n} \neq \{0\}$: \
    Note that there are $3^{3\brac{n+1}}$ ($2^{3\brac{n+1}}$) sign
    patterns $\sigma\in\{-1,0,1\}^{3n+3}$
    ($\sigma\in\{-1,1\}^{3n+3}$). A naive algorithm for computing
    nontrivial elements $\sigma$ of $\matMi{n}$ 
    from (\ref{eq:sign_condi}) 
    would therefore
    be to enumerate $\frac{1}{2}\, \big(3^{\brac{3n+3}} -1\big)$ ($
    2^{\brac{3n+2}}$) of these sign patterns and check feasibility
    of every linear system (\ref{eq:lin}) associated to one of those. (As
    feasibility of (\ref{eq:lin}) for an element $\sigma$ implies
    feasibility for $-\sigma$, one has to check only half of the sign
    patterns.) For the general problem of computing all sign patterns
    defined by two linear subspaces, a more advanced algorithm
    involving mixed integer linear programming is described in the
    thesis \cite{stj-09}. This software has aided the discovery of the
    sign patterns described below.
    \hbm
  \end{enumerate}
\end{remark}

Before turning to the elements of \matMi{n}, we conclude this section
by examining the linear system~(\ref{eq:lin}) for $n=2$:

\begin{example}[Steady states and corresponding rate constants for n=2
  via the system~(\ref{eq:lin})]\label{exa:rc}\mbox{} \rm  \\
  From (\ref{eq:def_M0}) -- (\ref{eq:def_Mn}) and (\ref{eq:def_Zi})
  we obtain \\ 
  \begin{equation*}
    \Mi{2} = 
    {\footnotesize \left[
        \begin{array}{rrr}
          -1 & -1 & 2 \\
          1 & 2 & -2 \\
          -1 & 0 & 1 \\\hline
          0 & 1 & 0 \\
          1 & 1 & -1 \\
          0 & 1 & 0 \\\hline
          0 & 0 & 1 \\
          1 & 0 & 0 \\
          0 & 0 & 1 \\
        \end{array}
      \right]}
    \ \mbox{ and } \
    \Zi{2} =  
    {\footnotesize \left[
        \begin{array}{ccc|ccc|ccc}
          1 & 0 & 0 & 1 & 0 & 0 & 1 & 0 & 0 \\
          0 & 0 & 1 & 0 & 0 & 1 & 0 & 0 & 1 \\
          0 & 1 & 0 & 1 & 1 & 1 & 1 & 1 & 1 \\
        \end{array}
      \right]}\ .
  \end{equation*}
  We observe that,  for the sign pattern
  \begin{displaymath}
    \sigma^{(2)}:=
    \left(\phantom{-}1,\phantom{-}1,\phantom{-}1\,
      |\phantom{-}1,-1,\phantom{-}1\, |-1,-1,-1\right)^T\, , 
  \end{displaymath}
  the columns of the matrices
  \begin{equation*}
    B= {\footnotesize \left[
        \begin{array}{rrr}
          1&0&0\\
          0&1&0\\
          2&1&1\\\hline
          1&1&0\\
          -1&0&-1\\
          1&1&0\\\hline
          0&0&-1\\
          -2&-1&-2\\
          0&0&-1
        \end{array}
      \right]}\ \mbox{ and } \
    C= {\footnotesize \left[
        \begin{array}{rrr|rrr}
          1&0&0&0&0&0\\
          1&1&0&1&0&1\\
          0&1&0&0&0&0\\\hline
          0&0&1&0&0&0\\
          0&0&0&-1&0&0\\
          0&0&0&0&1&0\\\hline
          -1&0&-1&0&0&0\\
          0&0&0&0&0&-1\\
          0&-1&0&0&-1&0  
        \end{array}
      \right]}
  \end{equation*}
  form a basis of $\im\brac{\Mi{2}}$ and of $\ker\brac{\Zi{2}}$ respectively.
  The expressions
  \begin{equation}\label{rev1.1}
    \mu=B\beta \in \, \im\brac{\Mi{2}}, \, \beta\in\Rp^3\, , \ \mbox{ and } \   
    s=C\gamma\in \ker\brac{\Zi{2}}\, \, ,\, \gamma\in\Rp^6\, ,
  \end{equation}
  define feasible vectors for  (\ref{eq:lin}) with the given sign pattern $\sigma^{(2)}$.
  Hence $\sigma^{(2)}$ defines a feasible system (\ref{eq:lin}):
  $\sigma^{(2)}\in\matMi{2}$.
  For a numerical example, we choose
  \begin{gather*}
    \beta_{1} = \ln(3/2),\, 
    \beta_{2} = \ln(4/3),\, 
    \beta_{3} = \frac{\ln2}{2},\\
    \gamma_{1} = 1,\, 
    \gamma_{2} = 1/2,\, 
    \gamma_{3} = 1,\, 
    \gamma_{4} = 1/4,\, 
    \gamma_{5} = 1,\, 
    \gamma_{6} = 1/4\, . 
  \end{gather*}
  Using these $\beta\in \Rp^3$ and $\gamma\in \Rp^6$, one obtains
  the vectors $\mu$ and $s$ given in Table~\ref{tab:mu_xi_ss}. Hence, by
  (\ref{eq:def_a_1}), (\ref{eq:def_a_2}) and (\ref{eq:def_b}), one
  arrives at the steady states $a$ and $b$ as given in
  Table~\ref{tab:mu_xi_ss}. The associated rate constant vectors $\kappa$ 
  are obtained by (\ref{eq:def_k}) for the value of $a$ from
  Table~\ref{tab:mu_xi_ss} and for free $\lambda\in\Rp^6$.
  Here we choose
  \begin{displaymath}
    \lambda = \left(
      30 \sqrt{2},\, 
      27,\, 
      30 \sqrt{2},\, 
      2+\sqrt{2},\, 
      2+\sqrt{2},\, 
      2+\sqrt{2}\, 
    \right)^T
  \end{displaymath}
  to obtain the rate constants given in
  Table~\ref{tab:rate_constants}. For the illustration in Fig.~\ref{fig:conti},
  a numerical continuation of the steady state
  $a$ has been performed using the rate constants from
  Table~\ref{tab:rate_constants}. The numerical continuation reveals
  that, for the rate constants of
  Table~\ref{tab:rate_constants} and the total concentrations $c=\Zi{2}a=\Zi{2}b$ defined
  by either $a$ or $b$, the system has three steady states, namely  $a$, $b$ and
  \begin{displaymath}
    \tilde{b} = \left(1.52..,
      1.01..,
      0.01..,
      0.12..,
      0.74..,
      0.12..,
      8.17..,
      4.39..,
      6.13..
      \, 
    \right)^T\, .
  \end{displaymath}
  A local
  stability analysis shows the exponential stability of $b$ and $\tilde{b}$.
  The steady state $a$ turns out to be a saddle point because of the presence of exactly
  one simple positive eigenvalue. Hence the system shows bistability as indicated in 
  Figure~\ref{fig:cartoon_multi}.  
  By the results of \cite{multi-001},
  these are all steady states of the double
  phosphorylation system once $\beta$, $\gamma$ and $\lambda$ are
  fixed.
  \hbm
\end{example}

\begin{table}
  \centering
    \begin{tabular}[c]{|*{4}{>{$}c<{$}|}}\hline
      \mu & s & a & b \\ \hline
      \ln\frac{3}{2} & 1 & 2 & 3 \\
      \ln\frac{4}{3} & 2 & 6 & 8 \\
      -\frac{\ln2}{2}+\ln6 & \frac{1}{2} &
      \frac{1}{34} \left(1+3 \sqrt{2}\right) & \frac{3}{34}
      \left(6+\sqrt{2}\right) \\ 
      \ln2 & 1 & 1 & 2 \\
      \ln\frac{\sqrt{2}}{3} & -\frac{1}{4} &
      \frac{3}{28} \left(3+\sqrt{2}\right) & \frac{1}{28} \left(2+3
        \sqrt{2}\right) \\ 
      \ln2 & 1 & 1 & 2 \\
      -\frac{\ln2}{2} & -2 & 2 \left(2+\sqrt{2}\right) & 2
      \left(1+\sqrt{2}\right) \\ 
      -\ln6 & -\frac{1}{4} & \frac{3}{10} & \frac{1}{20} \\
      -\frac{\ln2}{2} & -\frac{3}{2} & 3+\frac{3}{\sqrt{2}} &
      \frac{3}{2} \left(1+\sqrt{2}\right) \\ \hline
    \end{tabular}
    
    \vspace*{2mm}
    \caption{\label{tab:mu_xi_ss}Columns 1\&2: vectors $\mu$ and
        $s$ with $\sig{\mu} = \sig{s}$ satisfying the system
        (\ref{eq:lin}) (with $n=2$ and $\sigma^{(2)}$). Columns 3\&4:
        Steady state vectors $a$ and $b$ computed from $\mu$ and $s$
        via (\ref{eq:def_a_1}), (\ref{eq:def_a_2}) \&
        (\ref{eq:def_b}). Simple computation shows $c=\Zi{2}\, a =
        \Zi{2}\, b = \left(7+2 \sqrt{2},\frac{1}{34} \left(137+54
            \sqrt{2}\right),\frac{1}{140} \left(2187+505
            \sqrt{2}\right)\right)^T$.}
\end{table}

\begin{table}
  \centering
    \begin{tabular}[c]{|*{12}{>{$}c<{$}|}} \hline
      k_{1} & k_{2} & k_{3} &  l_1 &
      l_2 & l_3 & k_{4} & k_{5} & k_{6} &
      l_{4} & l_{5} & l_6 \\ \hline
      5 \sqrt{2} & 30 \sqrt{2} & 30 \sqrt{2} & 952 & 27 & 30 \sqrt{2} &
      \frac{4}{3} \left(4+\sqrt{2}\right) & \frac{1}{2} & \frac{1}{2} &
      \frac{40}{3} \left(4+5 \sqrt{2}\right) & \frac{2}{3} & \frac{2}{3}
      \\ \hline
    \end{tabular}
    
    \vspace*{2mm}
  \caption{\label{tab:rate_constants}One choice of rate
      constants such that $a$ and $b$ from Table~\ref{tab:mu_xi_ss}
      are steady states (i.e.\ 
      $\kappa=(k_1$, $k_2$, $k_3$, $l_1$, $l_2$, $l_3$, $k_4$, $k_5$,
      $k_6$, $l_4$, $l_5$, $l_6)^T$ with $\Si{2}\, \ri{2}{\kappa}{a} = \Si{2}\,
      \ri{2}{\kappa}{b} =0$).}
\end{table}

\begin{figure}
  \centering
  \includegraphics[width=.7\linewidth]{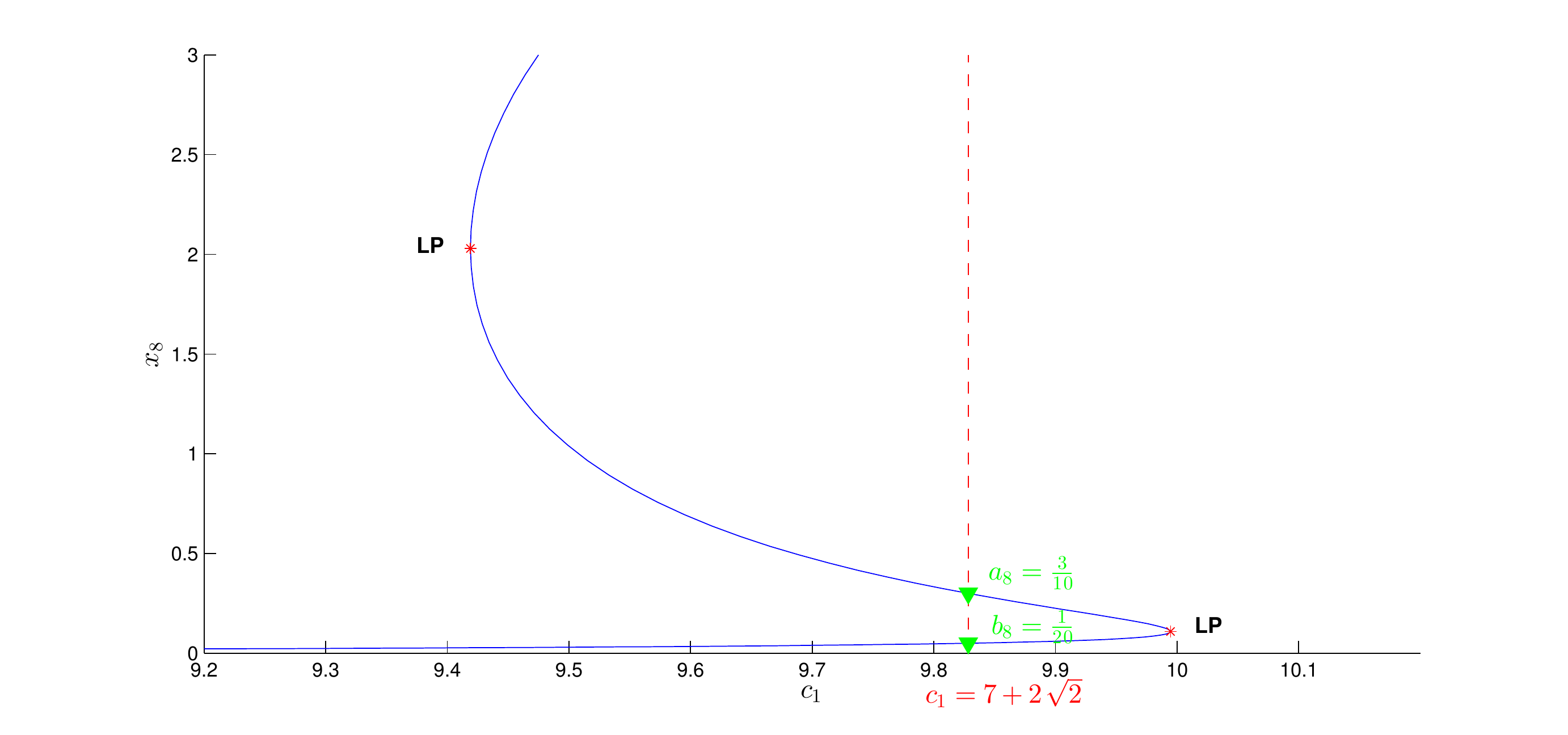}
  \caption{\label{fig:conti}Numerical continuation starting in $a$
    from Table~\ref{tab:rate_constants}. Here $x_8$ (i.e.\
    the concentration of $A_{2P}$) is plotted against $c_1$ (i.e.\ the
    total concentration of kinase $E_1$). Green triangles mark the
    components $a_8$ and $b_8$ of the steady state vectors $a$ and $b$
    from Table~\ref{tab:mu_xi_ss}, the total concentration $c_1$ at $a$
    (and $b$) is displayed in red. The  component $\tilde{b}_8$
      of the third steady state is close to 4.39.
    }
\end{figure}

\subsection{Sign patterns for feasibility of  system~(\ref{eq:lin})}
\label{sec:feasible_patterns}

Given a  $\sigma^{(n)}\in
\{-1,1\}^{3n+3}$, 
we first turn to the subproblem ${\Zi{n}}\, s = 0$  with
$\sig{s}=\sigma^{(n)}$ in ~(\ref{eq:lin}).  To this end, we consider sign patterns in column
and matrix form. With $(n+1)$ column vectors
$\tau_{(0)}=(\tau_{10},\tau_{20},\tau_{30})^T$, \dots,
$\tau_{(n)}=(\tau_{1n},\tau_{2n},\tau_{3n})^T$ from $\{-1,1\}^3$ we
associate  the sign pattern matrix 
\begin{equation}
  \label{taum}
  \tau=\left(\tau_{(0)}\, \cdots \, \tau_{(n)}\right)\in\{-1,1\}^{3\times (n+1)}\, .
\end{equation}
and the  sign pattern column
$\tau^{(n)}=\col(\tau)=\col(\tau_{(0)},...,\tau_{(n)})\in\{-1,1\}^{3n+3}$. 

\begin{lem}[
Feasibility of $\Zi{n}\, s = 0$ and $\sig{s} =\sigma^{(n)}\in \{-1,1\}^{3n+3}$]
\label{lem:S_feas1} \mbox{} \\
  For a given sign pattern matrix 
  $$\sigma =(\sigma_{(0)},...,\sigma_{(n)})=\left(\begin{array}{ccc}
      \sigma_{10}& ... & \sigma_{1n}\\
      \sigma_{20}& ... & \sigma_{2n}\\
      \sigma_{30}& ... & \sigma_{3n}
    \end{array}\right)
  \in \{-1,1\}^{3\times (n+1)}$$ 
  and its associated sign pattern vector
  $\sigma^{(n)}=col(\sigma_{(0)},...,\sigma_{(n)})\in \{-1,1\}^{3n+3}$, the subproblem
  \begin{equation}\label{subprob1}
    \Zi{n}\, s = 0 \ \mbox{ with } \ \sig{s} = \sigma^{(n)}
  \end{equation}
  in ~(\ref{eq:lin}) is not feasible if and only if $ \sigma$ 
  has one of the following three properties:
  \begin{enumerate}
  \item The elements of the first row are all of the same sign ${\sigma_{10}}$.
  \item The elements of the second row are all of the same sign ${\sigma_{20}}$ with
    ${\sigma_{20}}\sigma_{10}<0$ and ${\sigma_{20}}\sigma_{30}<0$.
  \item The elements of the third row are all of the same sign ${\sigma_{30}}$.
  \end{enumerate}
\end{lem}

\begin{proof}
  In analogy to $\sigma^{(n)}$, we use the notation
  $s=(s_{10},s_{20},s_{30},s_{11},s_{21},s_{31},\, ......\, ,
  s_{1n},s_{2n},s_{3n})^T$.
  The claim follows easily from the fact that $\Zi{n}\, s = 0$ is
  equivalent to the decoupled linear system
  \begin{displaymath}
    \sum_{j=0}^ns_{1j}=0, \quad \sum_{j=0}^ns_{3j}=0, \quad 
    \sum_{j=0}^ns_{2j}=s_{10}+s_{30}
  \end{displaymath}
  (cf. the definition of $\Zi{n}$ in eq.(\ref{eq:def_Zi})).
\end{proof}

\vspace{5mm}
We now turn to those sign patterns $\sigma^{(n)} \in \{-1,1\}^{3n+3}$,
where (\ref{subprob1}) is feasible and where, additionally, the second
subproblem 
\begin{equation}\label{subprob2}
  \mu = \Mi{n}\, \xi\ \mbox{ with } \ 
  \sig{\mu}=\sigma^{(n)}
\end{equation}
is feasible. In particular, we will show that (\ref{subprob1}) and
(\ref{subprob2}) are only feasible for sign patterns involving
\begin{subequations}
  \begin{align}
    \label{eq:building_blocks1a}
    \sigma_1 &:= \brac{+1,+1,+1}^T, & \sigma_2 &:= \brac{+1,-1,+1}^T, &
    \sigma_3 &:= \brac{-1,-1,-1}^T=-\sigma_1 \\
    \label{eq:building_blocks1b}
    \sigma_4 &:= \brac{+1,-1,-1}^T, & \sigma_5 &:=
    \brac{-1,+1,-1}^T=-\sigma_2\ , & \sigma_6 &:= \brac{-1,+1,+1}^T =
    -\sigma_4,
  \end{align}
  while it is infeasible for sign patterns involving
  \begin{equation}
    \label{eq:building_blocks2}
    \sigma_7:= \brac{+1,+1,-1}^T,\ \sigma_8:= \brac{-1,-1,+1}^T=
    -\sigma_7.
  \end{equation}
\end{subequations}
We use these $\sigma_j$ to encode  sign patterns
$\sigma^{(n)}\in\{-1,1\}^{3n+3}$. For example, for $n=2$ we encode
\begin{equation}\label{sigma2}
  \sigma^{(2)}:=
  \left(\phantom{-}1,\phantom{-}1,\phantom{-}1|\phantom{-}1,-1,\phantom{-}1|-1,-1,-1\right)^T\in
  \R^9 
\end{equation}
as the $3\times 3$ array
\begin{displaymath}
  \sigma = \left(\sigma_1\, \sigma_2\, \sigma_3\right)
\end{displaymath}
with the columns $\sigma_1$, $\sigma_2$ and $\sigma_3$.
We use the symbol $\rbrac{\tau}_i$ as shorthand  notation for the
$i$-fold concatenation of $\tau\in \{-1,1\}^3$, i.e., for the $3\times
i$--matrix
\begin{equation}
  \label{eq:def_concatenation}
  \rbrac{\tau}_i = (\, \underbrace{\tau\,
    \ldots\,\tau}_{\text{$i$ times}}\, )\, .
\end{equation}
For an example with $n=5$, 
\begin{displaymath}
  \sigma = \left(\rbrac{\sigma_3}_3\, \rbrac{\sigma_5}_2\, \sigma_1\right) =
  \left(\sigma_3\sigma_3\sigma_3\sigma_5\sigma_5\sigma_1\right)
\end{displaymath}
stands for the $18$-dimensional sign pattern
\begin{displaymath}\sigma^{(5)} = 
  \left(-1,-1,-1|-1,-1,-1|-1,-1,-1|-1,+1,-1|-1,+1,-1|+1,+1,+1\right)^T\, .
\end{displaymath}
In the following Theorem~\ref{theo:feasible_patterns} 
we characterize
the sign patterns $\sigma^{(n)}\in\{-1,1\}^{3n+3}$
that are feasible for \eqref{eq:lin} for $n\geq 2$.
In Section~\ref{sec:proof_theo_patterns} below we will prove this result.

\begin{theorem}\label{theo:feasible_patterns}
  Recall
  the linear system (\ref{eq:lin}) associated to the sign pattern vector
  $\sigma^{(n)}\in\{-1,1\}^{3n+3}$ and hence  to the equivalent sign pattern matrix $\sigma
  \in \{-1,1\}^{3\times (n+1)}$.
  For $n\geq 2$   the linear
  systems (\ref{eq:lin}) associated to $\sigma^{(n)}$  and  to
  $-\sigma^{(n)}$ are feasible  if and only if $\sigma$ is one of the
  following sign patterns \eqref{eq:s1} to \eqref{eq:s7}:
  \begin{itemize}
  \item For $n\geq 2$:
    \begin{align}      
      \label{eq:s1}
      \tag{$s_1$}
      \sigma &=\left(\sigma_2\, \rbrac{\sigma_3}_{i_1}\,
      \rbrac{\sigma_1}_{i_2}\right)\; \text{with $i_1$, $i_2\geq 1$ and
        $i_1+i_2 = n$} \\
	\label{eq:s2}
	\tag{$s_2$}
      \sigma &= \left(\sigma_3\, \rbrac{\sigma_5}_{i_1}\,
      \rbrac{\sigma_1}_{i_2}\right)\; \text{with $i_1$, $i_2\geq 1$ and
        $i_1+i_2 = n$} \\
      \label{eq:s3}
      \tag{$s_3$}
      \sigma &=\left( \sigma_4\, \rbrac{\sigma_5}_{i_1}\, \rbrac{\sigma_1}_{i_2}\right)\
      \text{with $i_1=1$, $i_2\geq 1$ and $i_1 +i_2= n$} \\
      \label{eq:s4}
      \tag{$s_4$}
      \sigma &=\left(\sigma_4\, \rbrac{\sigma_3}_{i_1}\,
      \rbrac{\sigma_1}_{i_2}\right)\; \text{with $i_1$, $i_2\geq 1$ and
        $i_1+i_2 = n$\, .}
    \end{align}
  \item Additionally for $n\geq 3$:
    \begin{align}
      \label{eq:s5}
      \tag{$s_5$}
      \sigma &=\left(\sigma_2\, \rbrac{\sigma_3}_{i_1}\,
      \rbrac{\sigma_2}_{i_2}\, \rbrac{\sigma_1}_{i_3}\right)\; \text{with $i_1$, 
        $i_2$, $i_3\geq 1$ and $i_1+i_2+i_3 = n$} \\
      \label{eq:s6}
      \tag{$s_6$}
      \sigma &=\left(\sigma_3\, \rbrac{\sigma_3}_{i_1}\,
      \rbrac{\sigma_5}_{i_2}\, \rbrac{\sigma_1}_{i_3}\right)\; \text{with
        $i_1$, $i_2$, $i_3\geq 1$ and $i_1+i_2+i_3 = n$} \\
      \label{eq:s7}
      \tag{$s_7$}
      \sigma &= \left(\sigma_4\, \rbrac{\sigma_3}_{i_1}\, \rbrac{\sigma_5}_{i_2}\,
      \rbrac{\sigma_1}_{i_3}\right)\ \text{with $i_1$ , $i_3\geq 1$, $i_2=1$ and $i_1+i_2+i_3 =
        n$\, .}
    \end{align}
  \end{itemize}
\end{theorem}
By Remark~\ref{rem:matro-condi}(A), $\sigma^{(n)}$  is
a nontrivial element of $\matMi{n}$ from \eqref{eq:sign_condi} if
and only if $\sigma^{(n)}$ is one of the sign patterns \eqref{eq:s1}
to \eqref{eq:s7} or its negative.
As an immediate consequence of Theorem~\ref{theo:feasible_patterns}
above we note the following:

\begin{remark}[$\Mi{n}\neq \{0\}$: multistationarity in
  $n$-site sequential distributive phosphorylation for $n\geq
  2$]\label{rem:multi_possible} \mbox{} \rm\\ 
  As a consequence of Theorem~\ref{theo:multi_linear_problem} and
  \ref{theo:feasible_patterns} we conclude that for every $n\geq 2$
  there exists a vector of rate constants $\kappa$ such that the
  system (\ref{eq:multistat_ode_x0}), (\ref{eq:multistat_ode_x1}) \&
  (\ref{eq:multistat_con_rel_x0_x1}) has at least two distinct positive
  solutions $a$ and $b$. Every element $\sigma$ described in
  Theorem~\ref{theo:feasible_patterns} defines a feasible linear
  system of the form (\ref{eq:lin}). 
  And solutions thereof
  parameterize the steady states $a$ and $b$ and the vectors
  $\kappa$ of rate constants via the equations~(\ref{eq:def_a_1}),
  (\ref{eq:def_a_2}) and (\ref{eq:def_k}) respectively 
  (cf.\ Example~\ref{exa:rc}).
  It will be
  interesting to systematically explore the corresponding regions in
  the rate constant space, especially with respect to biological
  plausibility and function.
\end{remark}  

\begin{remark}[$\matMi{1} = \{0\}$ excludes multistationarity for $n=1$
  (single phosphorylation)] \mbox{} \rm \\
  Any $\mu\in \im\brac{\Mi{1}}$ and any $s\in \ker\brac{\Zi{1}}$
  can be written as 
  \begin{displaymath}
    \mu=(\mu_1,\mu_4-\mu_1,\mu_3,\mu_4,\mu_4-\mu_1,\mu_4)^T \ \mbox{ and } \ 
    s=(s_1,s_2,s_3,-s_1,s_3-s_2+s_1,-s_3)^T
  \end{displaymath}
  respectively. We assume $\sign(\mu)=\sign(s)$. In case of
  $\mu_4=0$ we arrive at $\mu=0=s$. In case of $\mu_4\neq 0$ we
  assume without loss of generality $\mu_4>0$ and are led to the
  sign pattern $\sigi{1}=(-,+,-,+,+,+)^T$ 
  with sign pattern matrix
  \begin{displaymath}
    \sigma = \left(
      \begin{array}[c]{cc}
        - & + \\ + & + \\ - & +
      \end{array}
    \right).
  \end{displaymath}
  By Lemma~\ref{lem:S_feas1}
  these sign patterns cannot belong to an element of
  $\ker\brac{\Zi{1}}$ and we conclude $\matMi{1} = \{0\}$. 
\end{remark}

Finally we want to count the number of sign patterns of the
form~(\ref{eq:s1}) -- (\ref{eq:s7}). To demonstrate the reasoning
we first examine sign patterns of the form (\ref{eq:s2}) for $n=2$,
\ldots, 5:

\begin{example}[The sign pattern~(\ref{eq:s2})]\mbox{} \rm\\
  For $2\leq n\leq 5$ one obtains from the formula (\ref{eq:s2}) the
  following patterns:
  \vspace{2ex}
  \begin{center}
    \includegraphics[width=.7\linewidth]{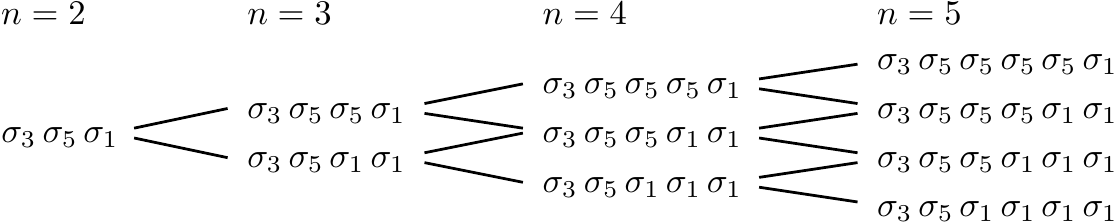}
  \end{center}
  \vspace{2ex}
  We  summarize for (\ref{eq:s2}) in case of $n=5$:

  \begin{itemize}
  \item  The integer $5$ can be partitioned in the sum of two integers
    $i_1$, $i_2$ in the following two ways: $5 = 1 + 4 =  3 + 2$.
    As we have to take the ordering of $i_1$ and $i_2$ into
    account (and hence, for example, \brac{4,1} differs from
    \brac{1,4}) we obtain the following pairs of integers:
    \begin{displaymath}
      \left\{ \brac{1,4},\; \brac{4,1},\; \brac{3,2},\; \brac{2,3}
      \right\}.
    \end{displaymath}
  \item This yields the sign patterns
    \begin{gather*}
      \sigma_3\, \rbrac{\sigma_5}_4\, \rbrac{\sigma_1}_1,\; \
      \sigma_3\, \rbrac{\sigma_5}_1\, \rbrac{\sigma_1}_4,\; \
      \sigma_3\, \rbrac{\sigma_5}_3\, \rbrac{\sigma_1}_2,\; \
      \sigma_3\, \rbrac{\sigma_5}_2\, \rbrac{\sigma_1}_3
    \end{gather*}
  \end{itemize}\hbm
\end{example}
Concerning the number of elements $\sigma^{(n)}\in\{-1,1\}^{3n+3}$ 
corresponding to (\ref{eq:s1}) -- (\ref{eq:s7}) we observe 
\begin{prop}
  \label{lem:nr_of_elements} \mbox{} \\
  For  $n\geq 3$ there exist $(n-1)(n+2)$  elements  in
  $\{-1,1\}^{3\times (n+1)}$ of the form (\ref{eq:s1}) --
  (\ref{eq:s7}).
\end{prop}
\begin{proof}
  We note that for $n\geq 3$ fixed formulae (\ref{eq:s5}) and
  (\ref{eq:s6}) each yield as many 
  elements as there are partitions of $n$ into three positive integers
  (taking order into account). Similarly (\ref{eq:s1}), (\ref{eq:s2})
  and (\ref{eq:s4}) each yield as many elements as there are
  partitions of $n$ into two positive integers (taking order into
  account), while (\ref{eq:s7}) yields as many elements as there are
  partitions of $n-1$ into two positive integers (taking order into
  account). Finally, (\ref{eq:s3}) yields one element.

  First we count the number of ways an integer can be partitioned into
  the sum of two positive integers (if order is taken into
  account). One has
  \begin{align*}
    n &= \brac{n-1} + 1 
    = \brac{n-2} +2 
    =\, \cdots \, 
    = \brac{n-\brac{n-1}} + \brac{n-1}.
  \end{align*}
  Note that there is no other way to partition $n$ into the sum of two
  positive integers. Hence we conclude that there is a total of $n-1$
  ways to partition $n$ into two positive integers. \\
  Second we count the number of partitions of $n$ into three positive
  integers. To this end we observe:
  \begin{align*}
    n &= \brac{n-2} + \brac{1 + 1} \\
    &= \brac{n-3} + \left\{\text{$2$ ways to write 3 as the sum of two
        positive integers}\right\} \\
    &\vdots \\
    &= \brac{n-j} + \left\{\text{$j-1$ ways to write $j$ as the sum of
        two positive integers}\right\} \\
    &\vdots \\
    &= \brac{n-\brac{n-1}} + \left\{\text{$n-2$ ways to write $n-1$ as
        the sum of two positive integers}\right\}
  \end{align*}
  Consequently there is a total of
  \begin{displaymath}
    \sum_{j=2}^{n-1}\, \brac{j-1}=\frac{1}{2}\, \brac{n-2}\, \brac{n-1}
  \end{displaymath}
  ways to partition $n$ into the sum of three positive
  integers. Now we can sum over the contribution of each of the
  formulae 
  \begin{displaymath}
    \underbrace{(n-2)(n-1)}_{\text{(\ref{eq:s5}) \&
        (\ref{eq:s6})}}
    + \underbrace{3\, \brac{n-1}}_{\text{(\ref{eq:s1}),
        (\ref{eq:s2}) \& (\ref{eq:s4})}}
    + \underbrace{(n-2)}_{\text{(\ref{eq:s7})}}
    \, + \, \underbrace{(\phantom{n} 1\phantom{2})}_{\text{(\ref{eq:s3})}}
    = (n-1)(n+2)\ .
  \end{displaymath}
\end{proof}

\vspace*{5mm}\section{Proof of Theorem~\ref{theo:poly_sol}}
\label{sec:proof_poly_sol}

To prove the Theorem~\ref{theo:poly_sol} we show the following implications:
\begin{displaymath}
  (A) \Rightarrow (D) \Rightarrow (C) \Rightarrow (B) \Rightarrow
  (A)
\end{displaymath}
  
\uu{(A) \Rightarrow (D):} \\ 
Suppose there exists $a$, $b\in\Rp^{3n+3}$ and
$\kappa\in\Rp^{3n+3}$ such that (\ref{eq:multistat_ode_x0}) \&
(\ref{eq:multistat_ode_x1}) hold. Then, by Theorem~4.3 of
\cite{ToricRN} one has
\begin{displaymath}
  \ln b - \ln a \in \im\brac{A^T},
\end{displaymath}
with
\begin{displaymath}
  A = \left[
    \begin{array}{ccc|ccc|ccc|c|ccc}
      1 & 0 & 0 & 1 & 1 & 1 & 2 & 2 & 2 & \dots & n & n & n \\
      1 & 0 & 1 & 1 & 0 & 1 & 1 & 0 & 1 & \dots & 1 & 0 & 1 \\
      0 & 1 & 0 & 1 & 1 & 1 & 1 & 1 & 1 & \dots & 1 & 1 & 1  \\
    \end{array}
  \right],
\end{displaymath}
where the columns of $A$ have been reordered to match the ordering
introduced in Table~\ref{tab:VarAssignment} (in the aforementioned
reference \cite{ToricRN}, the columns of $A$ correspond to the
exponents of $t_1$, $t_2$, $t_3$ in the statement of Theorem~4.3,
the matrix itself is displayed in the proof of said Theorem). 
\begin{displaymath}
  A^T=\Mi{n}Q \ \mbox{ for the nonsingular }\
  Q:=\left[\begin{array}{ccc}n&0&1\\1&1&1\\2&1&1\end{array}\right] 
\end{displaymath}
implying $\im\brac{\Mi{n}} = \im\brac{A^T}$
and hence
$(A) \Rightarrow (D)$ with $\mu:= \ln b - \ln a$.
  
\vspace*{1mm}\uu{(D) \Rightarrow (C):}\\
We first show that ${\Yi{n}}^T\, \Mi{n} = \PPi{n}$. To this end 
recall the definition of ${\Yi{n}}^T$ from eq.(\ref{eq:Y_direct})
\begin{displaymath}
  {\Yi{n}}^T = \left[
    \begin{array}{c|c|c|c|c|c|c}
      \Yo{1}^T & 0_{6\cdot 1 \times3} & \multirow{2}*{$0_{6\cdot
          2 \times 3}$} & \multirow{3}*{$0_{6\cdot 3 \times 3} $}&
      \phantom{0_{6\cdot 2 \times 3}} & \phantom{0_{6\cdot 2
          \times 3}} &  \multirow{5}*{$0_{6\cdot n \times 3}$} \\
      \cline{1-2}
      \multicolumn{2}{c|}{\Yo{2}^T} & & & & & \\ \cline{1-3}
      \multicolumn{3}{c|}{\Yo{3}^T} & & & \\
      \cline{1-4} \multicolumn{5}{r|}{\ddots} &  \\ \cline{1-5}
      \multicolumn{6}{c|}{\Yo{n-1}^T} &  \\ \hline
      \multicolumn{7}{c}{\Yo{n}^T}\\
    \end{array}
  \right].
\end{displaymath}
Now it is straightforward to compute that
\begin{displaymath}
  \Yo{i}^T\, \left[
    \begin{array}{c}
      M\brac{0,n} \\ \vdots \\ M\brac{i,n}
    \end{array}
  \right] = \left[
    \begin{array}{ccc}
      \uu{0} & \brac{-i+2}\, \uu{1} & \brac{i-1}\, \uu{1} 
    \end{array}
  \right] =: P\brac{i}\, \in \R^{6\times 3},
\end{displaymath}
where, as in Definition~\ref{def:Pi_and_Mi} we use the symbol
\uu{1} to denote a 6-vector filled with the number $1$ and,
accordingly, we use the symbol \uu{0} to denote a 6-vector 
filled with the number $0$. Thus one arrives at
\begin{displaymath}
  {\Yi{n}}^T\, \Mi{n} = \left[
    \begin{array}{c}
      P\brac{1} \\ \vdots \\ P\brac{n}
    \end{array}
  \right]
\end{displaymath}
Further note that
\begin{displaymath}
  P\brac{i}\, \begin{pmatrix}
    * \\ \xi_1 \\\xi_2
  \end{pmatrix}
  = \left[
    \begin{array}{cc}
      \brac{-i+2}\, \uu{1} & \brac{i-1}\, \uu{1} 
    \end{array}
  \right] \,
  \begin{pmatrix}
    \xi_1 \\\xi_2
  \end{pmatrix}
  = \Pi_0\brac{i}\, 
  \begin{pmatrix}
    \xi_1 \\\xi_2
  \end{pmatrix}
\end{displaymath}
(cf.\ Definition~\ref{def:Pi_and_Mi} and in particular
equation~(\ref{eq:def_Pi_0})).  Thus, given any nonzero
$\mu\in\im\brac{\Mi{n}}$ as
\begin{displaymath}
  \mu=\Mi{n}\, \begin{pmatrix}
    * \\ \xi_1 \\\xi_2
  \end{pmatrix}
\end{displaymath}
one arrives at
\begin{displaymath}
  {\Yi{n}}^T\, \Mi{n}\, \begin{pmatrix}
    * \\ \xi_1 \\\xi_2
  \end{pmatrix} =  \left[
    \begin{array}{c}
      \Pi_0\brac{1} \\ \vdots \\ \Pi_0\brac{n}
    \end{array}
  \right]\, 
  \begin{pmatrix}
    \xi_1 \\\xi_2
  \end{pmatrix} 
  = \PPi{n}\, \xi.
\end{displaymath}
Concerning $\nu$ and $\lambda$ as in (\ref{eq:lambda_free}) \&
(\ref{eq:def_nu_1}) -- (\ref{eq:def_nu_3}) we note that $\xi_1$,
$\xi_2$ from above together with any positive $\lambda$ can be
used to obtain $\nu$ as in (\ref{eq:def_nu_1}) --
(\ref{eq:def_nu_3}). Hence we conclude $(D) \Rightarrow (C)$.

\vspace*{1mm}\uu{(C) \Rightarrow (B):}\\
Let $\mu\in\R^{3n+3}$ with $\mu\neq 0$ and $\xi\in\R^2$ be
given, such that (\ref{eq:mu_lin}) holds. Further let $\nu$ and
$\lambda$ be given that satisfy (\ref{eq:lambda_free}) \&
(\ref{eq:def_nu_1}) -- (\ref{eq:def_nu_3}), that is
$\lambda\in\Rp^{3n}$ and 
\begin{displaymath}
  \nu_{3i} = \lambda_{3i}\, e^{\brac{2-i}\, \xi_1 + \brac{i-1}\,
    \xi_2}, \; \nu_{3i-1} = \lambda_{3i-1}\,
  \frac{\nu_{3i}}{\lambda_{3i}}\; \text{and}\;  \nu_{3i-2} =  
  \lambda_{3i-2}\, \frac{\nu_{3i}}{\lambda_{3i}}, i=1,\, \ldots,
  n.
\end{displaymath}
Observe that these $\nu$, $\lambda$ satisfy
\begin{displaymath}
  \ln\frac{\nu_{3i-2}}{\lambda_{3i-2}} = 
  \ln\frac{\nu_{3i-1}}{\lambda_{3i-1}} =
  \ln\frac{\nu_{3i}}{\lambda_{3i}} = \brac{2-i}\,
  \xi_1+\brac{i-1}\, \xi_2
\end{displaymath}
and hence
\begin{displaymath}
  \ln\frac{\nu_{3i-2}+\nu_{3i}}{\lambda_{3i-2}+\lambda_{3i}} =
  \ln\frac{\nu_{3i}}{\lambda_{3i}}\; \text{and}\; 
  \ln\frac{\nu_{3i-1}+\nu_{3i}}{\lambda_{3i-1}+\lambda_{3i}} =
  \ln\frac{\nu_{3i}}{\lambda_{3i}}.
\end{displaymath}
Define
\begin{displaymath}
  \nu_{\brac{i}} := \brac{\nu_{3i-2},\, \nu_{3i-1},\,
    \nu_{3i}}\; \text{and}\; \lambda_{\brac{i}} :=
  \brac{\lambda_{3i-2},\, \lambda_{3i-1},\, \lambda_{3i}}
\end{displaymath}
and let
\begin{displaymath}
  \nu = \col\brac{\nu_{\brac{1}},\, \ldots,\, \nu_{\brac{n}}}\;
  \text{and}\; \col\brac{\lambda_{\brac{1}},\, \ldots,\,
    \lambda_{\brac{n}}}.
\end{displaymath}
Recall the matrices \Ei{n} and $E$ from eqns.~(\ref{eq:Ei})
\& (\ref{eq:E1}) and observe that
\begin{gather*}
  \ln\frac{\Ei{n}\, \nu}{\Ei{n}\, \lambda} = \brac{\ln\frac{E\,
      \nu_{\brac{1}}}{E\, \lambda_{\brac{1}}},\, \ldots,\, \ln\frac{E\,
      \nu_{\brac{n}}}{E\, \lambda_{\brac{n}}}} \\
  \intertext{and}
  \ln\frac{E\, \nu_{\brac{i}}}{E\, \lambda_{\brac{i}}} = \brac{
    \ln\frac{\nu_{3i-2}+\nu_{3i}}{\lambda_{3i-2}+\lambda_{3i}},\,
    \ln\frac{\nu_{3i-2}}{\lambda_{3i_2}},\,
    \ln\frac{\nu_{3i}}{\lambda_{3i}},\,
    \ln\frac{\nu_{3i-1}+\nu_{3i}}{\lambda_{3i-1}+\lambda_{3i}},\,
    \ln\frac{\nu_{3i-1}}{\lambda_{3i-1}},\,
    \ln\frac{\nu_{3i}}{\lambda_{3i}}}.\\
  \intertext{Hence one has for $\nu$, $\lambda$ from above}
  \ln\frac{E\, \nu_{\brac{i}}}{E\, \lambda_{\brac{i}}} = \left[
    \begin{array}{cc}
      \brac{2-i}\, \uu{1} & \brac{i-1}\, \uu{1}
    \end{array}
  \right]
  \begin{pmatrix}
    \xi_1 \\ \xi_2
  \end{pmatrix}\\
  \intertext{and}
  \ln\frac{\Ei{n}\, \nu}{\Ei{n}\, \lambda} =
  \PPi{n}\, \xi.
\end{gather*}
Consequently, if there exist vectors $\mu\in\R^{3n+3}$, $\mu\neq
0$ and $\xi\in\R^2$ satisfying (\ref{eq:mu_lin}) and vectors $\nu$,
$\lambda$ satisfying (\ref{eq:lambda_free}), (\ref{eq:def_nu_1})
-- (\ref{eq:def_nu_3}), then $\mu$, $\nu$ and $\lambda$ also
satisfy  (\ref{eq:trans_eq}). Note that $\nu$ and $\lambda$
satisfying (\ref{eq:lambda_free}) \& (\ref{eq:def_nu_1}) --
(\ref{eq:def_nu_3}) are positive. Thus we conclude $(C)
\Rightarrow (B)$.

\vspace*{1mm}\uu{(B) \Rightarrow (A):} \\
Finally, assume $\mu\in\R^{3n+3}$ with $\mu\neq 0$ and
$\brac{\nu,\lambda}\in\Rp^{3n} \times \Rp^{3n}$ satisfy
(\ref{eq:trans_eq}). Fix $a\in\Rp^n$ and define
\begin{align*}
  b &:= \diag\brac{e^{\mu}}\, a \\
  \kappa &:= \diag\brac{\Phi^{\brac{n}}\brac{a^{-1}}}\, \Ei{n}\, \lambda.
\end{align*}
Then
\begin{displaymath}
  {\qquad}\Si{n}\, \diag\brac{\kappa}\, \Phi^{\brac{n}}\brac{a} =  \Si{n}\,
  \diag\brac{\Phi^{\brac{n}}\brac{a^{-1}}}\, \diag\brac{\Ei{n}\, \lambda}\,
  \Phi^{\brac{n}}\brac{a} = \Si{n}\, \Ei{n}\, \lambda = 0
\end{displaymath}
and (as, by~(\ref{eq:trans_eq}), ${\Yi{n}}^T\, \mu =
\ln\frac{\Ei{n}\, \nu}{\Ei{n}\, \lambda}$)
\begin{displaymath}
  \Si{n}\, \diag\brac{\kappa}\, \Phi^{\brac{n}}\brac{b} = \Si{n}\,
  \diag\brac{\kappa}\, \diag\brac{\Phi^{\brac{n}}\brac{a}}\,
  e^{{\Yi{n}}^T\, \mu} = \Si{n}\, \Ei{n}\, \nu = 0
\end{displaymath}
and therefore $(B) \Rightarrow (A)$ where we have used
equation~(\ref{eq:YT_mu}) of
Remark~\ref{rem:rkx_Phi_x}.\hbm

\vspace*{5mm}\section{Proof of Theorem~\ref{theo:feasible_patterns}}
\label{sec:proof_theo_patterns}

First, we observe that $\sig{\Mi{n}\, \xi} = \sigma^{(n)}$ implies
$\sig{\Mi{n}\, \brac{-\xi}} = -\sigma^{(n)}$ and that $\Zi{n}\, s = 0$
implies $\Zi{n}\, \brac{-s} = 0$. Hence, if the system (\ref{eq:lin})
associated to $\sigma^{(n)}$ is feasible for ($\xi$, $s$), then the system
(\ref{eq:lin}) associated to $-\sigma^{(n)}$ is feasible for ($-\xi$,
$-s$). Thus $\sigma^{(n)}\in\matMi{n}$ implies
$-\sigma^{(n)}\in\matMi{n}$.
Hence it suffices to show that the linear system (\ref{eq:lin}) is
feasible for elements $\sigma^{(n)}\in\{-1,1\}^{3n+3}$ with $\sigma_{3n}=+1$
if and only if $\sigma^{(n)}$ corresponds to one of the sign pattern
matrices in (\ref{eq:s1}) -- (\ref{eq:s7}).

\vspace*{2mm}
We now turn to the necessity of the sign patterns
(\ref{eq:s1}) -- (\ref{eq:s7}). To this end we parametrize the range
of the regular matrix $M\brac{0,n}$ from eq.~\eqref{eq:def_M0} by
\begin{equation}
w(0):=\begin{bmatrix}y\\x\\y-z\end{bmatrix}
=\begin{bmatrix}
        -1 & -n+1 & \phantom{-}n\\
        \phantom{-}1 & \phantom{-}n & -n\\
        -1 & -n+2 & \phantom{-}n-1
      \end{bmatrix}
      \begin{bmatrix}x+nz\\x+y\\x+y+z\end{bmatrix}\, .
\end{equation}
Consequently, one has for $M\brac{i,n}$, $i=1,....,n$, from eq.~\eqref{eq:def_Mi}
\begin{equation}
w(i)=\begin{bmatrix}w_1(i)\\w_2(i)\\w_3(i)\end{bmatrix}:=       \begin{bmatrix}x+iz+y-z\\x+iz\\x+iz+y-z\end{bmatrix}
=\begin{bmatrix}
0 & -i+2 & \phantom{-}i-1 \\
        1 & \phantom{-}n-i & -n+i\\
        0 & -i+2 & \phantom{-}i-1
        \end{bmatrix}
      \begin{bmatrix}x+nz\\x+y\\x+y+z\end{bmatrix}\, .
\end{equation}
Observe that the components  $w_2(i)$ and 
$w_1(i)=w_3(i)=w_2(i)+(y-z)$ are affine functions of $i$
so sign changes can be easily read off.
We now derive the necessary conditions for  $\sig{w(0)}$ to be one of
the four sign patterns 
$\sigma_1$, $\sigma_2$, $\sigma_4$ and $\sigma_7$ from
\eqref{eq:building_blocks1a} -- \eqref{eq:building_blocks2} under the
side condition that system~\eqref{subprob1} is not unfeasible
(cf. Lemma~\ref{lem:S_feas1}). The indices $i_k$ and $j_k$ to follow
will always be $\geq 1$.
\begin{enumerate}
\item $\sig{w(0)}=\sigma_2$, i.e., $y>0, x<0, y-z>0$:\\
  For $z\leq 0$, the $w_{2i}$ are all negative so that one has
  $\sigma_{20}=\sigma_{21}=\, \cdots \, =\sigma_{2n}=-1$ and
  $\sigma_{10}=\sigma_{30}=1$. By item (2) of Lemma~\ref{lem:S_feas1}, 
  $z\leq 0$ cannot lead to a feasible sign pattern. For $z>0$ the
  possible sign pattern matrices  are 
  \begin{displaymath}
    \left(\sigma_2,[\sigma_3]_{i_1}, [\sigma_1]_{i_2}\right),\ \
    \left(\sigma_2,[\sigma_3]_{j_1}, [\sigma_2]_{j_2},[\sigma_1]_{j_3}\right)
  \end{displaymath}
  for $i_1+i_2=n$ and $j_1+j_2+j_3=n$ where the $\sigma_1$ entries are
  necessary by item (2) of Lemma~\ref{lem:S_feas1}. Observe that the
  sign pattern $(\sigma_2,[\sigma_1]_n)$ is unfeasible by item (1) of
  Lemma~\ref{lem:S_feas1}.

\item $\sig{w(0)}=\sigma_3$, i.e., $y<0, x<0, y-z<0$:\\
  For $z\leq 0$, the components of all the $w(i)$ are of the same
  sign. By item (1) of Lemma~\ref{lem:S_feas1}, $z\leq 0$ cannot lead
  to a feasible sign pattern. For $z>0$, the $w(i)$ may generate the
  sign pattern matrices  
  $$(\sigma_3,[\sigma_5]_{i_1},[\sigma_1]_{i_2}), \ \
  (\sigma_3,[\sigma_3]_{j_1},[\sigma_5]_{j_2},[\sigma_1]_{j_3})$$ 
  with $i_1+i_2=n$ and $j_1+j_2+j_3=n$ where the $\sigma_1$ entries
  are necessary (otherwise the first rows do not offer a sign change). 
  The direct passage from $\sigma_3$ to $\sigma_1$ is obviously 
  impossible since $w_2(i)<0$ implies $w_1(i+1)=w_2(i)+y\leq w_2(i)<0$.

\item $\sig{w(0)}=\sigma_4$, i.e., $x<0, 0<y<z$:\\ 
  At first, possible sign pattern matrices are
  $$\left(\sigma_4,[\sigma_3]_{i_1}, [\sigma_1]_{i_2}\right),\ \
  \left(\sigma_4,[\sigma_5]_{i_1}, [\sigma_1]_{i_2}\right),\ \
  \left(\sigma_4,[\sigma_3]_{j_1}, [\sigma_5]_{j_2},[\sigma_1]_{j_3}\right)$$
  for $i_1+i_2=n$ and $j_1+j_2+j_3=n$ where the $\sigma_1$ entries are
  necessary by item (3) of Lemma~\ref{lem:S_feas1}. Observe that the
  sign pattern $(\sigma_4,[\sigma_1]_n)$ is unfeasible by item (1) of
  Lemma~\ref{lem:S_feas1}. Note that $\sigma_5$ can appear just once
  because of $w_{1,i+1}=w_{2i}+y\geq w_{2i}$.

\item $\sig{w(0)}=\sigma_7$, i.e., $x>0, 0<y<z$:\\
  Because of $z> 0$, the components  $w_{1i}$ are positive for
  $i=0,...,n$. By item (1) of Lemma~\ref{lem:S_feas1}, $\sigma_7$
  cannot lead to a feasible sign pattern.
\end{enumerate}

\vspace{3mm}
We finally establish the sufficiency of the sign patterns
(\ref{eq:s1}) -- (\ref{eq:s7}).
By Lemma~\ref{lem:S_feas1},
these $(s_k)$ are feasible sign patterns for the kernel of $\Zi{n}$.
In order to prove that they are realizable
we present
for each $(s_k)$, $k\in \{1,...,7\}$, a
vector 
\begin{subequations}\label{suffsign}
\begin{equation}\label{suffsign1}
\xi =col(x,x+y,x+y+z)\in \R^3 
\end{equation}
leading to the sign pattern $(s_k)$ for $\Mi{n}\xi$:
\begin{equation}\label{suffsign2}\begin{array}{cllc}
\mbox{For $(s_1)$:}\ &  x=i_2-\frac{3}{4},&y=-n+\frac{5}{4}, &z=1\, .\\[1mm]
\mbox{For $(s_2)$:}\ &  x=n-\frac{1}{4}, &y=-n-i_1+\frac{3}{4}, &z=1\, .\\[1mm]
\mbox{For $(s_3)$:}\ &  x=n-\frac{3}{4},  &y=-n+\frac{1}{4}, &z=1\, .\\[1mm]
\mbox{For $(s_4)$:}\ &  x=i_2-\frac{1}{4},& y=-n+\frac{3}{4}, &z=1\, .\\[1mm]
\mbox{For $(s_5)$:}\ &  x=i_3-\frac{3}{4}, &y=-n+i_2+\frac{5}{4}, &z=1\, .\\[1mm]
\mbox{For $(s_6)$:}\ &  x=i_2+i_3-\frac{1}{4},& y=-n-i_2+\frac{3}{4}, &z=1\, .\\[1mm]
\mbox{For $(s_7)$:}\ &  x=i_2+\frac{1}{4}, &y=-n+\frac{1}{4}, &z=1\, .
\end{array}
\end{equation}
\hbm
\end{subequations}

\section{Discussion}
\label{sec:discussion}

The main topic of this contribution is the existence of
multistationarity for $n$-site sequential distributive
phosphorylation. Via the equations (\ref{eq:def_k}) and
(\ref{eq:def_a_1}) -- (\ref{eq:def_b}) we were able to link the
existence of pairs of steady states ($a$, $b$) and a corresponding
vector of rate constants $\kappa$ to solutions of the linear systems
(\ref{eq:lin}). These systems are uniquely defined by sign patterns
$\sigi{n}\in\{-1,0,1\}^{3n+3}$ and hence the existence of a single
sign pattern defining a feasible system (\ref{eq:lin}) is
sufficient. Concerning sign patterns $\sigma\in\{-1,1\}^{3n+3}$ where
(\ref{eq:lin}) is feasible, Theorem~\ref{theo:feasible_patterns}
establishes four such sign patterns for $n=2$ and $(n-1)\, (n+2)$ for
$n\geq 3$.  

Unfortunately the biological interpretation of  elements $\sigma$
where (\ref{eq:lin}) is feasible is unclear (unlike in the case of 
multistationarity itself, that is usually associated with a specific
biological function). If different elements could be associated to
different biological functions, this would limit the number of
(different) biological functions a network can perform and hence be of
significant biological interest. Moreover, knowledge of all sign
patterns $\sigma$ defining feasible (\ref{eq:lin}) would then
correspond to knowledge of all biological functions a network is
capable of.

However, the following observations concerning pairs of steady states
($a$, $b$) can easily be made by inspection of the feasible sign patterns
(\ref{eq:s1}) -- (\ref{eq:s7}) and the structure of the matrices
\Mi{i} from (\ref{eq:def_Mi}): (i) every sign pattern ends with the
triplet $\brac{1,1,1}^T$ or $\brac{-1,-1,-1}^T$ and hence the last three
components of steady state $b$ are greater (smaller) than those of $a$
(recall that $\mu=\ln\frac{b_i}{a_i}$ and hence $\mu_i>0$ implies
$b_i>a_i$ and $\mu_i<0$ implies $b_i<a_i$). (ii) The first and the
third row of the matrices \Mi{i} are identical. Hence $\mu_{3i+1} =
\mu_{3i+3}$, $i=1$, \ldots, $n$. For a pair of steady states $(a,b)$
these correspond to the ratio of steady state concentrations of
kinase-substrate complex $A_{i-1\, p}\, E_1$ (in case of $\mu_{3i+1}$)
and phosphatase-substrate complex $A_{i_p}\, E_2$ (in case of
$\mu_{3i+3}$, cf. Table~\ref{tab:VarAssignment}).
Hence, for any pair
of steady states, the ratio of the steady state concentrations of
kinase-substrate complexes  equals that of phosphatase-substrate
complexes.

Finally we'd like to discuss two observations related to the formulae
(\ref{eq:def_k})  and (\ref{eq:def_a_1}) -- (\ref{eq:def_b}). 
Consider a sign pattern matrix $\sigma$ as in (\ref{eq:s1}) --
(\ref{eq:s7}) and let $\sigma^{(n)}$ be the corresponding sign
pattern vector. Pick vectors $\mu\in\im\brac{\Mi{n}}$,
$s\in\im\brac{\Si{n}}$ with $\sig{s} = \sig{\mu}=\sigma^{(n)}$. Let
($a^\sigma$, $b^\sigma$) be the pair of steady states defined by
$\mu$, $s$ via (\ref{eq:def_a_1}) -- (\ref{eq:def_b}). Let
($a^{-\sigma}$, $b^{-\sigma}$) be the pair of steady states defined
by $-\mu$ and $-s$. Then it is not difficult to see that
$a^{-\sigma}=b^{\sigma}$ and $b^{-\sigma} = a^{\sigma}$ (cf. a
similar discussion in \cite{cc-08}). Further note that a given pair
($a$, $b$), obtained via (\ref{eq:def_a_1}) -- (\ref{eq:def_b}), is
a pair of steady states for all vectors of rate constants defined by
the formula~(\ref{eq:def_k}), where $\lambda\in\Rp^{3n}$ is free.

The formula~(\ref{eq:def_k}) for the vectors of rate constants may
be interesting for two reasons: first, it allows to replace the $6n$
rate constants $k_i$, $l_i$ by the $3n$ \lq coordinates\rq{}
$\lambda_i$ and hence a reduction of the degrees of freedom in the
dynamical system (\ref{eq:ode_def}). Second, it  means that
multistationarity is robust with respect to perturbations within the
image of the map (\ref{eq:def_k}) and fragile with respect to
transversal perturbations. In a future publication we will
explore the image of this map with the aim of establishing parameter
values for multistationarity in biologically meaningful domains.

\subsection{Acknowledgments}

KH and CC acknowledge financial support from the International Max
Planck Research School in Magdeburg and the Research Center 'Dynamic
Systems' of the Ministry of Education of Saxony-Anhalt, respectively.
Finally, we'd like to thank the diligent reviewers for their
valuable suggestions.


\end{document}